\documentclass{article}
\usepackage{geometry}
\usepackage{amssymb, amsmath, dsfont} 
\usepackage{graphicx}

\usepackage{amsthm}
\usepackage{natbib}
\usepackage[colorlinks,citecolor=blue]{hyperref}

\theoremstyle{plain}
\newtheorem{theorem}{Theorem}
\newtheorem{lemma}[theorem]{Lemma}
\newtheorem{corollary}[theorem]{Corollary}
\newtheorem{proposition}[theorem]{Proposition}

\theoremstyle{definition}
\newtheorem{definition}[theorem]{Definition}

\newtheorem{condition}[theorem]{Condition}

\theoremstyle{remark}
\newtheorem{remark}[theorem]{Remark}
\usepackage{enumerate}

\newcommand{\Ind}{\mathds{1}}
\newcommand{\ind}{\Ind}
\newcommand{\E}{\mathds{E}}
\newcommand{\var}{\mathds{V}ar}
\newcommand{\Var}{\var}
\newcommand{\Prob}{\mathds{P}}
\newcommand{\R}{\mathds{R}}

\newcommand\LRT{\operatorname{LR}}
\newcommand{\given}{\middle|}

\begin{document}

\title{A Reproducing-Kernel-Hilbert-Space log-rank test for the two-sample problem}

\author{
Tamara Fern\'andez \thanks{Gatsby Computational Neuroscience Unit, University College London. 
{\tt t.a.fernandez@ucl.ac.uk}}
\and Nicol\'as Rivera  \thanks{Computer Laboratory, University of Cambridge. 
{\tt nicolas.rivera@cl.cam.ac.uk}}
}

\maketitle

\abstract{
Weighted log-rank tests are arguably the most widely used tests by practitioners for the two-sample problem in the context of right-censored data. Many approaches have been considered to make weighted log-rank tests more robust against a broader family of alternatives, among them, considering linear combinations of weighted log-rank tests, and taking the maximum among a finite collection of them. In this paper, we propose as test statistic the supremum of a collection of (potentially infinite) weight-indexed log-rank tests where the index space is the unit ball in a reproducing kernel Hilbert space (RKHS). By using  some desirable properties of RKHSs we provide an exact and simple evaluation of the test statistic and establish connections with previous tests in the literature. Additionally, we show that for a special family of RKHSs, the proposed test is omnibus. We finalise by performing an empirical evaluation of the proposed methodology and show an application to a real data scenario. Our theoretical results are proved using techniques for double integrals with respect to martingales that may be of independent interest.

\textbf{Key words:} Survival Analysis, Right-Censored Data, Reproducing Kernel Hilbert Space, Log-rank Test, two-sample tests.}

\section{Introduction}
Two-sample testing is a classical problem in the context of survival data. For instance, in a clinical trial, two-sample tests can be used to compare different treatments when the survival times of patients are censored. Within the context of right-censored data, the classical log-rank test, first introduced by \citet{mantel1966evaluation} and \citet{peto1972asymptotically}, is the most widely used test among practitioners. A well-known property of the classical log-rank test is that it is the most powerful test under the assumption of proportional-hazard alternatives. This result can be deduced by noticing that the log-rank test statistic coincides with the score test statistic when the true cumulative hazard function belongs to the model $\left\{\Lambda_\theta,\theta\in\Theta\right\}$, where $\Lambda_\theta(t)=\int_0^t e^{\theta}d\Lambda_0(x)$, $\Lambda_0$ denotes the cumulative hazard function under the null hypothesis, and $\theta\in\Theta\subseteq \R$ is the parameter of the model. While the classical log-rank test is  optimal for proportional-hazard alternatives, it can have a substandard behaviour when the true cumulative hazard function cannot be expressed in terms of  $\{\Lambda_{\theta},\theta\in\Theta\}$. 

In order to broaden the power of the classical log-rank test to other  families of alternatives, researchers have introduced and studied different variants of weighted log-rank tests \citep{tarone1977distribution, gill1980censoring, harrington1982class, Bagdonavivius2010, andersen2012statistical}. We refer the reader to Chapter 7 of   \citet{klein2006survival} for a general discussion and comparison of different weighted log-rank tests approaches. In the simplest case, each weighted log-rank test with weight function $\omega$ can be written as a score test for alternatives in the model $\left\{\Lambda_{\theta,\omega}(t),\theta\in\Theta\right\}$, where $\Lambda_{\theta,\omega}(t)=\int_0^t e^{\theta\omega(F_0(x))}d\Lambda_0(x)$ and $\omega$ is fixed. Then, similarly to the above, it can be deduced that each weighted log-rank test is the most powerful test for the null hypothesis $H_0:\Lambda_1=\Lambda_0$ under the assumption that $\Lambda_1$ can be expressed in terms of $\left\{\Lambda_{\theta,\omega},\theta\in\Theta\right\}$, for some $\theta\neq 0$. If the true cumulative hazard function cannot be expressed in terms of this parametric model, there are no guarantees at all for the behaviour of the weighted log-rank test. Indeed, it may happen that we observe pathological cases in which the test has zero asymptotic power for specific alternatives. This is the case of the classical log-rank test, which is recovered by choosing $\omega =1$, in the setting of crossing-hazard alternatives.

In an attempt to overcome the previously described limitations of weighted log-rank tests,  researchers have considered two natural approaches to improve their performance: the first approach considered is to learn the weight function from the data, which defines a weighted log-rank test with an adaptive weight function, and the second approach consists on combining several log-rank tests into a single test statistic. Both approaches avoid making the strong parametric assumption that a particular model is true. The first approach has been discussed by  \citet{lai1991rank}, \citet{yang2005combining} and \citet{yang2010improved}, and the second approach, which is the focus of this paper, has been studied by several authors. Particularly,  \citet{tarone1981distribution}, \citet{kossler2010max}, and \citet{gares2015omnibus}  combined weighted log-rank tests by considering the maximum of a finite collection of them, and \citet{kosorok1999versatility} proposed the supremum of an infinite collection of weight-indexed log-rank tests, with weights belonging to a particular space of functions. The latter approach is very close to what we propose in this paper, however, the test statistic proposed by \citet{kosorok1999versatility} lacks of an analytically tractable expression, which forces the authors to rely on a Monte-Carlo approximation of it, whereas our approach, by an appropriate selection of the space of functions, obtains a simple expression for our test statistic, leading to a simple testing procedure.  

While the focus of this paper are weighted log-rank test and the combination of them, there are other test-approaches for the two-sample problems that are worth mentioning. One example of these approaches are the weighted Kaplan-Meier estimators \citep{Pepe1989WeightedKM}, which are parametrised on a weight function, and they can also be combined into one single statistic \citep{Shen2001Maximum}. A different approach is the so-called empirical likelihood approach \citep{Mai2016Empirical}, which leads to very interesting test statistics that can also be combined into a single statistic \citep{Bathke2009Combined}. Somewhat classical options are the Cramer von-Mises and the Kolmovorov-Smirnov test statistics, that have been deeply studied in many settings, including Survival Analysis \citep{koziol1976cramer,Koziol82omnibus}.

In this paper, we consider the supremum of a potentially infinite collection of weighted log-rank tests with weights belonging to the unit ball of a reproducing kernel Hilbert space of functions. Then, we propose to test the null hypothesis $H_0:\Lambda_1=\Lambda_0$, by checking if $\sup_{\omega\in\mathcal{H}:\|\omega\|_{\mathcal{H}}\leq 1}\text{LR}_n(\omega)\approx 0$, where $\text{LR}_n(\omega)$ denotes a weighted log-rank test with weight function $\omega$, and $\mathcal{H}$ is a reproducing kernel Hilbert space of functions. Clearly, if $\Lambda_1$ can be expressed in terms of  $\{\Lambda_{\theta,\omega},\theta\in\Theta,\omega\in B_1(\mathcal{H})\}$ for some $\theta\neq 0$ and some $\omega\in B_1(\mathcal{H})$ (where $B_1(\mathcal{H})$ denotes the unit ball of $\mathcal{H}$), our test statistic should be strictly greater than zero, in which case we decide against the null hypothesis. Also, notice that, the larger the space of functions $\mathcal{H}$, the more models we test at the same time. 

Choosing weights $\omega$ in the unit ball of a reproducing kernel Hilbert space $\mathcal H$ is the key step in our testing procedure, as we will show that, by doing so, our test statistic $\sup_{\omega\in\mathcal{H}:\|\omega\|_{\mathcal{H}}\leq 1}\text{LR}_n(\omega)$ can be evaluated exactly  by a straightforward computation. This result follows from the good properties of reproducing kernel Hilbert spaces. Particularly, we will use two properties of reproducing kernel Hilbert spaces: i) they are uniquely characterised by a kernel function $K$,  i.e., a symmetric and positive-definite function, and ii) they satisfy the so-called \emph{reproducing kernel} property,  which states that $\omega(x)=\langle\omega, K(x,\cdot)\rangle_{\mathcal H}$ holds for any function $\omega\in\mathcal{H}$ and any point $x$. By using this property we will show that $\sup_{\omega\in \mathcal H: \|\omega\|_{\mathcal H}}\text{LR}_n(\omega)^2= \sum_{i=1}^n\sum_{j=1}^n K(x_i,x_j)c_ic_j$, for some points $x_i$ and constants $c_j$ (depending on the data). This type of kernel estimators have been previously addressed by \citet{Neuhaus1987Local} in connection with infinite sums of linear rank estimators (such as the log-rank) in the uncensored setting.  Since, in practice, our test statistic is a quadratic form based on the kernel $K$, our method can be described just in terms of $K$ avoiding the need to describe a Hilbert space. Describing kernel functions with interesting properties and useful interpretation has been a topic of extensive research, especially in the Machine Learning community. A compendium of popular kernels used in applications can be found in \citet{sousa2010}.

While we can easily establish connections between our method and approaches based on weighted log-rank tests, we will show that an alternative interpretation of our test statistic allows us to connect our approach with the \emph{kernel mean embedding} testing approach. In particular, we will show that our test statistic $\sup_{\omega\in\mathcal{H}:\|\omega\|_{\mathcal{H}}\leq 1}\text{LR}_n(\omega)$  can be alternatively written as the norm of the difference of particular embeddings of our two samples into $\mathcal H$. The idea of comparing probability distributions/datasets by embedding them into a reproducing kernel Hilbert space of functions has been extensively studied in the uncensored  setting by researchers in Statistics and Machine Learning \citep{Berlinet2004Reproducing, smola2007hilbert}, however, it seems this idea has yet to percolate into the Survival Analysis community and, up to the best of our knowledge, it has only been considered by \citet{fernandez2018maximum} for the goodness-of-fit problem under right-censored data. In the uncensored case, two-sample tests were proposed and studied by \citet{gretton2012kernel}.  In this work, the authors embed two different samples into a reproducing kernel Hilbert space by considering the so-called \emph{``kernel mean embedding"} of the empirical distribution of each sample. Then, they compute the difference between these two samples by computing the distance (induced by the norm) of the embeddings into $\mathcal{H}$. This idea can be straightforwardly applied to right-censored data by considering the kernel mean embedding of the Kaplan-Meier estimator instead of the empirical distribution, which leads to test statistics that are Kaplan-Meier $V$-statistics. These type of statistics were studied by \cite{bose2002asymptotic} and by \citet{FerRiv2020}, and include as particular cases well-known statistics such as the Cramer von-Mises test statistic. As it was pointed out by \citet{FerRiv2020}  such an approach is not suitable for censored data as Kaplan-Meier $V$-statistics are only reliable when the amount of censored observations is rather small when compared to the total amount of data, and even in that case, the test statistic may not be data-efficient.

In this work, we study several asymptotic properties of our test statistic under both, the null and alternative hypothesis. Under the null hypothesis, we find the limiting distribution of our test statistic by approximating it by a degenerate V-statistic.  Under the alternative hypothesis, we prove that under reasonable conditions, our test has asymptotic power tending to one. While the asymptotic null distribution is known, in most cases, its parameters are hard to compute. Thus, in order to implement our test,  we propose a Wild Bootstrap approximation of the null distribution and prove that the Wild Bootstrap statistic converges in distribution to our test statistic under the null hypothesis, giving us a correct Type-I error. Finally, we show that particular instances of our test statistic recover some existing testing approaches studied in the literature, which suggests our testing procedure is a natural generalisation of them. Examples of tests which our testing procedure recovers are: weighted log-rank tests, Pearson-type  tests \citep{akritas1988} and projection-type tests \citep{Brendel2013weighted, ditzhaus2018more}.

Besides theoretical guarantees, we provide an extensive empirical evaluation of our method in two important data scenarios:  proportional hazard functions and  time-dependent hazard functions, including Weibull and periodic hazard alternatives.  In our experiments we demonstrate that our method has a good performance in a wide range of problems, which include sample sizes from small to large, and different censoring percentages. Our experiments show that finite-dimensional reproducing kernel Hilbert spaces tend to have an overall better performance in problems with fewer observations or a simpler hazard function. While more complex reproducing kernel Hilbert spaces (i.e., infinite-dimensional spaces) still show a good performance in these simpler problems, our experiments show that they are better suited for larger datasets and more complex hazard functions. We also provide a real-data evaluation of our method, trying different kernel functions, and we compare the results with those obtained by  \citet{ditzhaus2018more}.

The structure of the paper is as follows. In Section \ref{sec:LogRank}, we introduce some standard notation used in Survival Analysis and describe  weighted log-rank tests. In Section \ref{sec:RKHSLR} we introduce the essential background knowledge about reproducing kernel Hilbert spaces (RKHS) and formally introduce our test statistic. In Section \ref{sec:limitdistri} we study asymptotic properties of our test. Later, in Section \ref{sec:Wild}, we proceed to explain how to implement a Wild Bootstrap approach. Sections \ref{sec:Simu} and \ref{sec:Real} are devoted to empirical studies using simulated data and real data, respectively.

\section{Survival analysis background}\label{sec:LogRank}
\subsection{General notation}
We establish some general notation that will be used throughout the paper. We denote $\R_+ = (0,\infty)$. Let $f:\R_+\to\R_+$ be an arbitrary right-continuous function, then we define $f(x-) = \lim_{h \to 0}f(x-|h|)$. In this work we make use of standard  asymptotic notation \citep{janson2011probability}, e.g., $O_p$, $o_p$, $\Theta_p$, etc., with respect to the number of sample points $n$. In order to avoid large parentheses, we write $X = O_p(1)Y$ instead of $X= O_p(Y)$, especially if the expression for $Y$ is very long. Given a sequence of stochastic processes $(W_n(x): x\in \mathcal X)$, depending on the number of observations $n$, and a function $f(x)\geq 0$, we say that $W_n(x) = O_p(1)f(x)$ uniformly for all $x\in A_n$, if and only if $\sup_{x\in A_n}|W_n(x)|/f(x) = O_p(1)$, where $A_n\subseteq \mathcal X$ is a set that may depend on $n$, and we use the convention $0/0 = 0$. Lastly, in this work the integral symbol $\int_{a}^b$ means integration over $(a,b]$ unless we state otherwise.

\subsection{Right-censored data}
Our data corresponds to right-censored observations belonging to two groups/classes, namely group $0$ and group 1. We assume that the total number of observations is $n$, and that $n = n_0+n_1$, where $n_c$ denotes the number of observations in group $c\in\{0,1\}$. We use $[n]$ to denote the set $\{1,\ldots, n\}$. For the sake of asymptotic analysis, we assume that we have no vanishing groups, that is, $n_0/n \to \eta $ and  $n_1/n \to (1-\eta) $ with $\eta\in(0,1)$ as $n$ tends to infinity.

Right-censored datasets are commonly observed in triples $((X_i, \Delta_i, c_i))_{i\in[n]}$, where $X_i=\min\{T_i,C_i\}$ is an \emph{observed time}, defined as the minimum between a \emph{survival time} of interest $T_i$ and a \emph{censoring time} $C_i$, the variable
$\Delta_i=\ind_{\{X_i= T_i\}}$ is an indicator of whether we actually observe the survival time of interest, and $c_i \in \{0,1\}$ is an associated covariate that denotes the group membership of the $i$-th observation. In this work we assume that all triples are mutually independent and that the triples have the same distribution within each group. Additionally, we assume \emph{non-informative censoring}, meaning that, given the covariate $c_i$,  the survival time $T_i$ is independent of the censoring time $C_i$.
 
We denote by $F_c$, $G_c$, and $H_c$, the respective conditional distribution functions of the random variables $T_i$, $C_i$, and $X_i$, given that the covariate $c_i$ takes the value $c\in\{0,1\}$. Notice that $1-H_c=(1-F_c)(1-G_c)$ due to the non-informative censoring assumption. For simplicity of exposition, we assume that the distribution functions $F_c$ and $G_c$ are continuous distributions on $\R_+$, however, all the methods of this paper can be extended to general probability distributions on $\R$, as our results are based on counting processes arguments that have been developed in full generality. We denote by $S_c(x)=1-F_c(x)$ and by $\Lambda_c(x)=\int_{(0,x]}S_c(s)^{-1}dF_c(s)$, the so-called \emph{survival} and \emph{cumulative hazard} functions associated with $F_c$, respectively,  and  we denote by $\widehat S_c$ the  Kaplan-Meier estimator of $S_c$, and by $\widehat \Lambda_c$ the Nelson-Aalen estimator of $ \Lambda_c$. Notice that for the computation of $\widehat S_c$ and $\widehat\Lambda_c$ we only use triples $(X_i,\Delta_i,c_i)$ satisfying $c_i = c$.

While in principle we assume that the covariates $c_i$ are deterministic, it is also possible to consider them as  independent and identically distributed (i.i.d.) samples from a Bernoulli distribution with mean $\eta$ (notice that in this case $n_0$ is a random variable satisfying $n_0/n \to \eta$ almost surely). These two possible assumptions originate two models: the \emph{deterministic covariates} model and the \emph{random covariate}s model. For our analysis, it will be convenient to use the random covariates model as, under this assumption, the triples $((X_i,\Delta_i,c_i))_{i\in[n]}$ are i.i.d. which simplifies the asymptotic analysis, nevertheless, as we will show in Appendix~\ref{sec:equivalence}, both models are asymptotically equivalent. Also, in practice, there is no difference between the two models when implementing our testing procedure.

In this paper, we use the notion of pooled data and pooled distributions. The pooled data corresponds to the original data after ignoring/removing the covariates $(c_i)_{i\in[n]}$, that is, the pooled data is $((X_i,\Delta_i))_{i\in[n]}$. We use the term pooled distributions to refer to the distributions of a randomly selected data-point. Particularly, under the random covariates model, the pooled distributions associated with the observed times $(X_i)_{i\in[n]}$ and the survival times $(T_i)_{i\in[n]}$, are the marginal distributions of $X_i$ and $T_i$, respectively, given by $H = \eta H_0 + (1-\eta)H_1$ and $F = \eta F_0 + (1-\eta)F_1$. Under the deterministic covariates model, the pooled  distributions associated with the observed times and the survival times correspond to $(n_0/n)H_0 + (n_1/n)H_1$ and $(n_0/n)F_0 + (n_1/n)F_1$, respectively. Notice that, since $(n_0/n)H_0 + (n_1/n)H_1\to H$ and $(n_0/n)F_0 + (n_1/n)F_1\to F$ as $n$ grows to infinity, the pooled distributions of both models are asymptotically equivalent. Additionally, we denote by $S(x)=1-F(x)$ and by $\Lambda(x) = \int_{(0,x]} S(s)^{-1}dF(s)$ the survival and cumulative hazard functions associated with the pooled distribution $F$, and we use the Kaplan-Meier estimator $\widehat{S}$ and the Nelson-Aalen estimator $\widehat{\Lambda}$ to approximate $S$ and $\Lambda$, respectively. Notice that $\widehat{S}$ and $\widehat{\Lambda}$ are computed using the pooled-data. 

\subsection{Counting processes}\label{sub:countingprocess}
In this work, we use the standard counting processes notation used in Survival Analysis. We define the \emph{individual, class and pooled} counting processes by $N^i(x)=\Delta_{i}\ind_{\{X_i\leq x\}}$ for $i\in[n]$, $N_c(x)=\sum_{i=1}^{n}\ind_{\{c_i=c\}}N^i(x)$ for $c\in\{0,1\}$, and $N(x)=\sum_{i=1}^n N^i(x)=N_0(x)+N_1(x)$, respectively. Similarly, we define the individual, class and pooled risk functions by $Y^i(x)=\ind_{\{X_{i}\geq x\}}$, $Y_c(x)=\sum_{i=1}^{n}\ind_{\{c_i=c\}}Y^i(x)$ and $Y(x)=\sum_{i=1}^n Y^i(x)= Y_0(x)+Y_1(x)$, respectively. By using the previous notation, we write the Nelson-Aalen estimator of $d\Lambda_c$  as $d\widehat \Lambda_c(x) = dN_c(x)/Y_c(x)$, and the pooled Nelson-Aalen estimator as $d\widehat \Lambda(x) = dN(x)/Y(x)$. To avoid ambiguities we assume that $0/0 = 0$, e.g., $d\widehat \Lambda(x) = 0$ if $Y(x) = 0$. Finally,  we define the process $L(x)=Y_0(x)Y_1(x)/Y(x)$, which appears frequently in our results, and we define $\tau_n = \max\{X_i: i \in [n]\}$ and $\tau=\sup\{x:1-H(x)>0\}$. Observe that $\tau_n\to\tau$ a.s. and that $N^i(x)=N^i(\tau_n)$, $N_c(x)=N_c(\tau_n)$ and $N(x)=N(\tau_n)$ for any $x\geq \tau_n$. Also, notice that all our observed times, $X_i$, are less than $\tau$ as they are generated from continuous distributions.

We assume that all random variables are defined on a common filtrated probability space $(\Omega,\mathcal{F},(\mathcal{F}_{x})_{x\geq 0},\Prob)$, where the sigma-algebra $\mathcal F_x$ is generated by 
\begin{align*}
\left\{\ind_{\{X_{i}\leq s,\Delta_{i}=0\}},\ind_{\{X_{i}\leq s,\Delta_{i}=1\}}:0\leq s\leq x,i \in [n]\right\},
\end{align*}
and the $\Prob$-null sets of $\mathcal{F}$. Under the deterministic covariates model, we define the individual $(\mathcal{F}_x)$-martingales associated to each data point, $(X_i,\Delta_i,c_i)$, as $M^i(x)=N^i(x)-\int_{(0,x]}Y^i(s)d\Lambda_{c_i}(s)$. Similarly, under the deterministic covariates model, we define the class and pooled martingales by $M_c(x)=\sum_{i=1}^{n}\ind_{\{c_i = c\}}M^i(x)$  and $M(x)=\sum_{i=1}^n M^i(x)= M_0(x)+M_1(x)$, respectively, where $c \in \{0,1\}$. Notice that under the random covariates model, $M^{i}(x)$, $M_c(x)$ and $M(x)$ are $(\mathcal{F}_x)$-martingales when conditioning on $(c_i)_{i\in[n]}$ (alternatively, we can include the sigma algebra associated with $(c_i)_{i\in[n]}$ to generate $\mathcal F_x$). We denote by $\langle M^i \rangle(x)$ and $[M^i](x)$, respectively, the predictable and quadratic variation processes associated with $M^i(x)$, which, in the context of continuous survival and censoring times, are given by $d\langle M^i \rangle(x) = Y^i(x) d\Lambda_{c_i}(x)$ and $d[M^i](x)= dN^i(x)$. For further information about counting processes and their applications in Survival Analysis we refer the reader to \citet{Flemming91}.

Notice that, since we consider continuous survival times, we can estimate $F(x)$ by either $\widehat{F}(x)$ or $\widehat F(x-)$. The main advantage of using the latter estimator is that $(\widehat F(x-))_{x\geq 0}$ is left-continuous, and thus, it is a predictable process with respect to the filtration $(\mathcal F_x)_{x\geq 0}$.

\subsection{The log-rank estimator}
The weighted log-rank statistic is defined by
\begin{eqnarray}\label{eqn:log-rank}
\LRT_n(\omega) = \frac{n}{n_0n_1}\int_0^{\tau_n} \omega(\widehat F(x-))L(x)\left(d\widehat \Lambda_0(x)-d\widehat \Lambda_1(x)\right),
\end{eqnarray}
where the function $\omega: [0,1) \to \R$ is referred to as \emph{weight function}. In the previous equation recall that $L(x)=Y_0(x)Y_1(x)/Y(x)$, $\widehat F$ is the pooled Kaplan-Meier estimator of the pooled distribution function $F$, and $\widehat \Lambda_0$ and $\widehat \Lambda_1$ are the class Nelson-Aalen estimators of the cumulative hazard functions $\Lambda_0$ and $\Lambda_1$, respectively. 

\citet{gill1980censoring} studied the asymptotic behaviour of the weighted log-rank statistic, $\LRT_n(\omega)$, under the null hypothesis $H_0:F_0 = F_1$, obtaining that
\begin{align}
\sqrt{\frac{n_0n_1}{n}}\LRT_n(\omega)&\overset{\mathcal{D}}{\to}N(0,\sigma^2),\label{eqn:nullasympnormal}
\end{align} 
where 
\begin{align*}
\sigma^2&=\int_0^\tau w(F_0(x))^2\frac{(1-G_0(x))(1-G_1(x))}{\eta(1-G_0(x))+(1-\eta)(1-G_1(x))}dF_0(x).
\end{align*}
Notice that, even under the null hypothesis, the censoring distributions $G_0$ and $G_1$ are not necessarily equal, hence the expression given for $\sigma^2$ cannot be simplified. Also, in most cases, the censoring distributions are unknown, and thus it is not possible to evaluate $\sigma^2$ to find rejection regions. To carry-out a testing procedure, we can use the asymptotic null distribution given in Equation \eqref{eqn:nullasympnormal}, replacing the asymptotic variance $\sigma^2$ by $\widehat \sigma^2=n/(n_0n_1)\int_0^{\tau_n} \omega(\widehat F(x-))^2 L(x) d\widehat \Lambda (x)$, which is an unbiased estimator of $\sigma^2$ under the null hypothesis. We refer the reader to \cite{gill1980censoring} for more details.

It is a well-known fact that weighted log-rank tests relate to score tests through the choice of a specific family of alternatives. Let
\begin{align}
\Lambda(x;\theta,\omega)&=\int_0^xe^{\theta\omega(F_0(s))}d\Lambda_0(s),\quad\theta\in\Theta,\label{eqn:familyAternal}
\end{align}
be a parametric model (indexed by $\theta$) for the cumulative hazard function, where $\omega:[0,1)\to\R$ is a fixed continuous function, and $\Theta$ is chosen as an open subset of $\R$ containing $0$. Then, under the assumption that $\Lambda_1(x)=\Lambda(x;\theta_1,\omega)$ for some $\theta_1 \in \Theta$,  testing the null hypothesis, $H_0:\Lambda_0 = \Lambda_1$, is equivalent to test $H_0:\theta = 0$. 

We can test $H_0:\theta=0$ by using a score test. The score statistic is computed in terms of the score function defined by $U(\theta)=\partial/\partial\theta\log L_n(\theta;\omega)$, where $L_n(\theta;\omega)$ is the likelihood function under the model of Equation $\eqref{eqn:familyAternal}$. For the Goodness-of-Fit problem, where $F_0$ and $\Lambda_0$ are known, the score statistic is given by
\begin{align*}
U(0)&=\int_0^{\tau_n} w(F_0(x)) Y_1(x) \left(d \Lambda_0(x)-d \widehat \Lambda_1(x)\right).
\end{align*}
In the Two-Sample problem, $\Lambda_0$ and $F_0$ are unknown, but they can be estimated using the \textbf{pooled} Nelson-Aalen and Kaplan-Meier estimators, obtaining 
\begin{align*}
\widehat U(0)&=\int_0^{\tau_n} w(\widehat F(x-)) L(x) \left(d \widehat\Lambda(x)-d \widehat \Lambda_1(x)\right).
\end{align*}
Then, a simple comparison shows that $n/(n_0n_1)\widehat U(0)=\LRT_n(\omega)$, from which we deduce the relationship between the weighted log-rank test defined in Equation \eqref{eqn:log-rank} and the score test associated to the parametric model defined in Equation~\eqref{eqn:familyAternal}. By the Neyman-Pearson Lemma, we deduce that the weighted log-rank test is the most powerful test for small departures from the null, that is, when $\theta \to 0$. Unfortunately, if the model in Equation~\eqref{eqn:familyAternal} is misspecified, little can be said about the performance of the test, indeed, it is well-known that, in some cases, the weighted log-rank statistic yields a test with asymptotically zero-power.

\section{An RKHS approach to the log-rank test}\label{sec:RKHSLR}

A standard approach to broaden the power of weighted log-rank tests, in order to achieve a more robust behaviour across a larger class of alternatives, is to combine several weighted log-rank tests into a single statistic. In this way, if one of the weight functions completely fails to differentiate between the null and alternative hypotheses, we can still rely on the remaining weight functions to help us to discriminate, and thus increase the overall power of our testing procedure. Two interesting questions that arise from this approach are: How do we combine these weighted log-rank tests? Which and how many weight functions do we need to choose? 

Selecting and combining weight functions efficiently for the problem at hand is very difficult as, in most cases, it requires that the user analyses the data in advance to select appropriate weight functions,  e.g., it is usual to check if the hazard functions cross or/and if they show late/early departures. Searching for important features in the data may be very time-consuming and, moreover, it is always possible that there are features that do not translate into an appropriate weight function. Also, it might happen that the user overlooks some important relations.

Instead of relying on the expertise of a user to identify important features, we can consider as many weight functions as possible to account for the heterogeneity in the data.  We can take this approach to the extreme by considering a potentially infinite family of weights.  While this approach solves the problem of choosing weight functions, naively choosing a particular family of weights  will lead to tests that, i) are hard to calibrate as we will need several data points to reach the correct Type-I error, ii) do not have an analytically tractable closed form for the test statistic, leading to tests that iii) are computationally expensive.

In order to overcome the previous difficulties, we consider weight functions belonging to the unit ball of a reproducing kernel Hilbert space $\mathcal H$, and propose as test statistic $Z_n$, which is defined by
\begin{align}
Z_n&=\sup_{\omega\in\mathcal{H},\|\omega\|^2_{\mathcal{H}}\leq 1}\LRT_n(\omega)^2.\label{eqn:supStat}
\end{align}
While a priori there is not a good reason to choose this particular family of weights,  we will show that it has very nice features that translate into desirable properties of our testing procedure; among others, the test statistic has a simple closed form, allowing simple computations as well as an economic Wild Bootstrap implementation.

\subsection{Reproducing kernel Hilbert spaces}\label{sec:RHKSintro}

A reproducing kernel Hilbert space (RKHS) is a Hilbert space of functions $\omega:[0,1)\to \R$ satisfying that the evaluation $\omega \to \omega(x)$ is continuous for every fixed $x \in [0,1)$ (it is worth mentioning that we can replace $[0,1)$ by any space).  By the Riesz representation theorem, for any $x\in[0,1)$, there exists a unique element $K_x \in \mathcal H$ such that $\omega(x) = \langle K_x, \omega\rangle_{\mathcal H}$ for all $\omega \in \mathcal H$, which is known as the \emph{reproducing kernel property}. Since $K_x \in \mathcal H$ for all $x\in[0,1)$, $K_x(y) = K_y(x) = \langle K_x, K_y \rangle_{\mathcal H}$ holds for any $x,y \in [0,1)$. Then, this result allows us to define the so-called \emph{reproducing kernel} $K:[0,1)^2\to  \R$ as
\begin{align}\label{eqn:repKernel}
K(x,y) =  \langle K_x, K_y \rangle_{\mathcal H}.
\end{align}
From now on, in order to ease the notation, we write $K(x,\cdot)$ instead of $K_x(\cdot)$, even though the former induces a slight abuse of notation.
 
For every RKHS $\mathcal H$ with inner product $\langle \cdot, \cdot \rangle_{\mathcal H}$ there exists a unique symmetric positive-definite reproducing kernel $K:[0,1)^2\to\R$ satisfying Equation~\eqref{eqn:repKernel}. Conversely, by the Moore-Aronszajn Theorem \citep{Aronszajn1950}, for any  symmetric positive-definite kernel function $K:[0,1)^2\to \R$, there exists a unique RKHS $\mathcal H$ for which $K$ is its reproducing kernel.

Given a finite signed measure $\mu$ on $[0,1)$, we define the kernel mean embedding of $\mu$ into $\mathcal H$ as
\begin{align}\label{eqn:defiMKE}
\mu \to  \int_{[0,1)} K(y,\cdot)\mu(dy) \in \mathcal H,
\end{align}
where the previous integral has to be understood as a Bochner integral. A sufficient condition to guarantee the existence of such an embedding is that  $\int_{[0,1)} \sqrt{K(x,x)}|\mu|(dx)<\infty$ \citep[Lemma 3]{gretton2012kernel}. A kernel $K$ is said to be \emph{characteristic} if the mean kernel embedding, defined in Equation~\eqref{eqn:defiMKE}, is injective on the space of probability distributions (i.e., distinct probability measures are embedded as different elements of $\mathcal H$). Furthermore, a continuous kernel $K$ is said to be \emph{$c$-universal} if the mean kernel embedding is injective on the set of finite signed  measures. Clearly, a continuous $c$-universal kernel is characteristic. Most of the standard kernels used in applications are $c$-universal, e.g., the Gaussian kernel, $K(x,y) = \exp\{-(x-y)^2/\sigma^2\}$, and the Ornstein–Uhlenbeck kernel, $K(x,y) = \exp\{-|x-y|/|\sigma|\}$.  For more information about characteristic and $c$-universal kernels we refer the reader to \cite{fukumizu2009}, \cite{maudet2017}, \cite{simon2018}
and references therein.

\subsection{An RKHS log-rank test}
Recall that 
\begin{align*}
Z_n=\sup_{\omega\in\mathcal{H},\|\omega\|^2_{\mathcal{H}}\leq 1}\LRT_n(\omega)^2,
\end{align*}
where $\mathcal H$ is an RKHS of functions $\omega:[0,1) \to \R$ with reproducing kernel $K:[0,1)^2\to\R$. In this section, we show that it is possible to find a closed-form expression for our test statistic, $Z_n$, in terms of the reproducing kernel $K$ associated with $\mathcal{H}$.  This result is formally stated in the following theorem:
\begin{theorem}\label{thm:ZndoubleIntegral} 
\begin{align}
Z_n&=\left(\frac{n}{n_0n_1}\right)^2\int_0^{\tau_n}\int_0^{\tau_n}K(\widehat F(x-),\widehat F(y-))L(x)L(y)\left(d\widehat{\Lambda}_0(x)-d\widehat{\Lambda}_1(x)\right)\left(d\widehat{\Lambda}_0(y)-d\widehat{\Lambda}_1(y)\right).\label{eqn:ExplicitZ}
\end{align}
\end{theorem}
By using the definition of the class Nelson-Aalen estimators, Equation~\eqref{eqn:ExplicitZ} can be rewritten as 
\begin{align}
Z_n&=\left(\frac{n}{n_0n_1}\right)^2\sum_{i=1}^n \sum_{j=1}^n K(\widehat F(X_i-), \widehat F(X_j-)) L(X_i)L(X_j) (-1)^{c_i+c_j}\frac{\Delta_i}{Y_{c_i}(X_i)}\frac{\Delta_j}{Y_{c_j}(X_j)},\label{eqn:ZnComputationImplementation}
\end{align}
which is a quadratic form with matrix-representation given by
\begin{align}\label{eqn:Z_nMatrix}
Z_n= \left(\frac{n}{n_0n_1}\right)^2\mathbf{V}^{\intercal} \widehat{ \mathbf K}\mathbf{V},
\end{align}
where $\mathbf{V} \in \R^n$ is defined by $\mathbf{V}_j = L(X_j)(-1)^{c_j}\Delta_j/Y_{c_j}(X_j)$  and  $\widehat{\mathbf{K}}\in \R^n\times \R^n$ is the matrix defined by $\widehat{\mathbf{K}}_{ij} = K(\widehat F(X_i),\widehat F(X_j))$. 

In order to prove Theorem~\ref{thm:ZndoubleIntegral}, we introduce an alternative embedding of the data, which is different to the kernel mean embedding of Equation \eqref{eqn:defiMKE}. Given a measure  $\nu$ (not necessarily a probability measure), we define the  embedding $\phi$ of a finite signed measure $\nu$ into $\mathcal H$ by 
\begin{align*}
\phi_{\nu}(\cdot)&=\int_{0}^{\tau} K(F(y),\cdot)\nu(dy).
\end{align*}
In practice, since the pooled distribution $F$ is unknown, we replace it by the pooled Kaplan-Meier estimator $\widehat{F}(x-)$ (recall that the survival time are continuous). Then, by using the previous definition, we introduce 
\begin{align}
\phi_{0}^n(\cdot) = \frac{n}{n_0n_1}\int_0^{\tau_n} K( \widehat F(y-),\cdot) L(y)d\widehat{\Lambda}_0(y),\quad\text{and}\quad\phi_{1}^n(\cdot) = \frac{n}{n_0n_1}\int_0^{\tau_n} K(\widehat F(y-),\cdot) L(y) d\widehat\Lambda_1(y),\label{eqn:Embeddingsn}
\end{align}
which are the corresponding embeddings of two empirical measures, $\nu^n_0$ and $\nu^n_1$, into $\mathcal{H}$, where $\nu^n_0$ and $\nu^n_1$ are defined for any $s<t$ by
\begin{align}
\nu_0^n((s,t]) = \frac{n}{n_0n_1}\int_{s}^t L(x)d\widehat{\Lambda}_0(x),\quad&\text{and}\quad \nu^n_1((s,t]) = \frac{n}{n_0n_1}\int_{s}^t L(x)d\widehat{\Lambda}_1(x),\label{eqn:measuresn}
\end{align}
respectively. Notice that $\phi_0^n$ and $\phi_1^n$ are always well-defined,  meaning that $\phi_0^n\in\mathcal{H}$ and $\phi_1^n\in\mathcal{H}$, since they are just finite sums of elements of $\mathcal H$ (observe that $K(\widehat F(y-), \cdot) \in \mathcal H$ for any fixed $y$). 

By using the previous embeddings, we obtain an inner product representation of the weighted log-rank statistic, which will be used in the proof of Theorem \ref{thm:ZndoubleIntegral}.
\begin{lemma}[Log-rank representation]\label{lemma:LRembed}
For any $\omega\in\mathcal{H}$,
\begin{align*}
\LRT_n(\omega)&=\left\langle w, \phi_{0}^n-\phi_{1}^n \right\rangle_{\mathcal H},
\end{align*}
and then
\begin{align*}
    Z_n = \|\phi_0^n-\phi_1^n\|_{\mathcal H}^2.
\end{align*}
\end{lemma}
\begin{proof}
Since $\widehat F(x-) \in [0,1)$ for any $x\leq \tau_n$, the reproducing property yields $\omega(\widehat F(x-)) = \langle \omega(\cdot), K(\widehat F(x-),\cdot) \rangle_{\mathcal H}$. Then, by the linearity of the inner  product and integration,
\begin{align*}
\LRT_n(\omega) &= \frac{n}{n_0n_1}\int_0^{\tau_n} \omega(\widehat F(x-))L(x)\left(d\widehat \Lambda_0(x)-d\widehat \Lambda_1(x)\right)\nonumber\\
&= \frac{n}{n_0n_1}\int_0^{\tau_n} \langle w, K(\widehat F(x-),\cdot) \rangle_{\mathcal H}L(x) \left(d\widehat\Lambda_0(x) - d\widehat\Lambda_1(x)\right)\nonumber\\
&=  \left\langle w, \frac{n}{n_0n_1}\int_0^{\tau_n} K(\widehat F(x-),\cdot)L(x) \left(d\widehat\Lambda_0(x) - d\widehat\Lambda_1(x)\right) \right\rangle_{\mathcal H}\nonumber\\
&= \left\langle w, \phi^n_{0}-\phi^n_{1} \right\rangle_{\mathcal H}.
\end{align*}
Finally, by taking supremum over the unit ball, we have that 
\begin{align}
    Z_n&=\sup_{\omega\in\mathcal{H},\|\omega\|^2_{\mathcal{H}}\leq 1}\LRT_n(\omega)^2=\sup_{\omega\in\mathcal{H},\|\omega\|^2_{\mathcal{H}}\leq 1} \left\langle w, \phi^n_{0}-\phi^n_{1} \right\rangle_{\mathcal H}^2 = \|\phi_0^n-\phi_1^n\|_{\mathcal H}^2,
\end{align}
where the last equality holds since $\mathcal H$ is a Hilbert space.
\end{proof}

\begin{proof}[Proof of Theorem~\ref{thm:ZndoubleIntegral}]
By Equation \eqref{eqn:Embeddingsn}, we have  $(\phi_0^n-\phi_1^n)(\cdot) = \frac{n}{n_0n_1}\int_{0}^{\tau_n} K(\widehat F(x-),\cdot)L(x)( d\widehat \Lambda_0(x)-d\widehat \Lambda_1(x))$. Then, by using the linearity of the inner product and integration, and the reproducing kernel property, we deduce
\begin{align*}
&\|\phi_0^n-\phi_1^n\|_{\mathcal{H}}^2 \nonumber\\
&\quad= \left\langle \frac{n}{n_0n_1}\int_{0}^{\tau_n} K(\widehat F(x-), \cdot)L(x)\left(d\widehat \Lambda_0(x)-d\widehat \Lambda_1(x)\right),\frac{n}{n_0n_1}\int_{0}^{\tau_n} K(\widehat F(y-), \cdot)L(y)\left(d\widehat \Lambda_0(y)-d\widehat \Lambda_1(y)\right)\right \rangle_{\mathcal{H}}\\
&\quad=\left(\frac{n}{n_0n_1}\right)^2 \int_0^{\tau_n}\int_0^{\tau_n} \left \langle K(\widehat F(x-), \cdot), K(\widehat F(y-), \cdot)\right\rangle_{\mathcal H} L(x)L(y)\left( d\widehat \Lambda_0(x)-d\widehat \Lambda_1(x)\right)\left( d\widehat \Lambda_0(y)-d\widehat \Lambda_1(y)\right)\\
&\quad=\left(\frac{n}{n_0n_1}\right)^2 \int_0^{\tau_n}\int_0^{\tau_n} K(\widehat F(x-), \widehat F(y-)) L(x)L(y)\left(d\widehat \Lambda_0(x)-d\widehat \Lambda_1(x)\right)\left(d\widehat \Lambda_0(y)-d\widehat \Lambda_1(y)\right).
\end{align*}
\end{proof}

\subsection{Recovering existing tests}\label{sec:Recovering}
Our testing approach is based on fixing an RKHS of functions $\mathcal{H}$, which is done by fixing a kernel $K$. We show that for specific choices of the reproducing kernel $K$, we recover some previously known tests.
\subsubsection{Weighted log-rank tests}\label{sec:weilogranktest}
In order to recover the classical weighted log-rank test $\LRT_n(\omega)$, we set $K(x,y) = \omega(x)\omega(y)$. Notice that $K$ is symmetric and non-negative definite for any $\omega$. Then, by evaluating $Z_n$ using Equation \eqref{eqn:ExplicitZ}, we obtain
\begin{align*}
Z_n = \LRT_n(\omega)^2.
\end{align*}
\subsubsection{Pearson-type tests}\label{sec:pearsontest}
Consider the partition of the interval $(0,1]$ given by $I_j = (j/k,(j+1)/k]$, where $j \in \{0,\ldots, k-1\}$ and $k>0$ is an integer. We recover Pearson-type testing approaches by choosing the kernel $K(x,y) = \sum_{j=0}^{k-1} \omega(x)\omega(y)\ind_{I_j \times I_j}(x,y),$ where $\omega: [0,1) \to \R$ is a given weight function. Then, by evaluating $Z_n$ using Equation \eqref{eqn:ExplicitZ}, we get
\begin{align*}
Z_n = \sum_{j=0}^{k-1} \LRT_n\left(\omega \ind_{\{I_j\}}\right)^2=\sum_{j=0}^{k-1}\left(\int_{I_j}\omega(\widehat F(x-))L(x)\left(d\widehat{\Lambda}_0(x)-d\widehat{\Lambda}_1(x)\right)\right)^2.
\end{align*}
Particularly, if we choose $\omega \equiv 1$, we recover 
\begin{align}\label{eqn:pearsonTest}
Z_n &=\sum_{j=0}^{k-1}\left(\sum_{X_{i}\in I_j}\frac{Y_1(X_{i})}{Y(X_{i})}\Delta_{i}\ind_{\{c_i = 0\}}-\sum_{X_{i}\in I_j}\frac{Y_0(X_{i})}{Y(X_{i})}\Delta_{i}\ind_{\{c_i=1\}}\right)^2,
\end{align}
which compares observed frequencies on $I_j$ between the two groups. We can go further and consider the normalised version of Pearson-type tests by choosing the kernel as 
\begin{align*}
K(x,y) = \sum_{j=0}^{k-1}\frac{\omega(x)\omega(y)}{\sigma_j(\omega)^2}\ind_{I_j \times I_j}(x,y),
\end{align*}
where $\sigma_j(\omega)^2=\frac{n}{n_0 n_1}\int_{I_j} \omega(\widehat F(x-))^2L(x) d\widehat \Lambda (x)$. Notice that in this case $ K(x,y)$ is a random kernel, but it converges to a deterministic one since $\sigma_j(\omega)$ converges almost surely to a constant when $n$ tends to infinity.

\subsubsection{Projection-type test}\label{sec:projectiontest}
Projection-type tests were introduced by \citet{Brendel2013weighted}, and recently revisited by  \citet{ditzhaus2018more}. Consider a finite number of weight functions $w_1,\ldots, w_k$ such that $w_i\circ F \in L^2(\nu)$ for all $i\in\{1,\ldots, k\}$,  where $\nu$ is the measure given by 
\begin{align*}
\nu(dx)& =  \frac{(1-H_0(x))(1-H_1(x))}{1-H(x)}d\Lambda(x).
\end{align*}
The projection-type testing approach considers the following test statistic
\begin{align*}
\sum_{i=1}^k \LRT(\tilde w_i)^2,
\end{align*}
where $\{\tilde w_i\}_{i=1}^k$ is an orthonormal base of the subspace $U$ of $L^2(\nu)$ generated by $w_1\circ F,\ldots, w_k \circ F$. We refer the reader to \citet{Brendel2013weighted} for a detailed explanation. \citet{Brendel2013weighted} and \citet{ditzhaus2018more}  recommend using weights with some meaning,  for instance,  $\omega(x)=x-1/2$  is used to detect a cross between two hazard functions around the median of the pooled survival time distribution, and  $\omega(x)=x^p (1-x)^q$ is used to detect early and/or late differences between the hazard functions, depending on the parameters $p$ and $q$. Nevertheless, observe that, in terms of projections, the meaning of the functions does not matter as, for instance, projecting over the subspace generated by $\{1,x-1/2, x(1-x), x^2(1-x)\}$ is the same as projecting over the subspace generated by $\{1,x,x^2, x^3\}$. 
 
Our approach can be seen as a natural generalisation of the previous method. Indeed, we can recover the previous test statistic by choosing the kernel 
\begin{align*}
K(x,y) = \sum_{i=1}^k \tilde w_i(x) \tilde w_i(y).    
\end{align*}
To compute the orthonormal basis $\{\tilde \omega_i\}_{i=1}^k$, we just need to compute the $k\times k$ matrix $\boldsymbol P$, where $\boldsymbol P_{ij} = \langle w_i, w_i\rangle_{ L^2(\nu)} = \int_0^\tau \omega_i(F(x))\omega_j(F(x)) d\nu(x)$, and then we set $\tilde \omega_i(\cdot) = \sum_{j}\boldsymbol P^{-1}_{ij}w_j(\cdot)$. Notice that the matrix $\boldsymbol P$ may not have an inverse, meaning that the functions $\omega_i$ are not linearly independent in $L^2(\nu)$. In such a case we compute its Moore-Penrose pseudo-inverse. Since, in practice, we do not have enough information to compute $\boldsymbol P$, we can estimate it (under the null hypothesis) by $ \widehat{\boldsymbol P}_{ij} = n/(n_0 n_1)\int_0^\tau \omega_i(\widehat F(x-))\omega_j(\widehat F(x-)) L(x) d\widehat \Lambda(x)$. Finally, observe that this procedure generates a random kernel that converges to a deterministic one when $n$ tends to infinity.

\section{Asymptotic Properties}\label{sec:limitdistri}
In this section we study some asymptotic properties of our test statistic $Z_n$. Before introducing our main results, we state some standard and very reasonable conditions that are needed in our proofs. We first state some smoothness, and  moment conditions regarding the kernel $K$.

\begin{condition}\label{con:conditions}
Let $K:[0,1)^2\to\R$ be a continuous kernel, and let $X,X',Y,Y'$ be independent random variables such that $X,X' \sim F_0$ and $Y,Y' \sim F_1$. Then, we assume that
\begin{enumerate}[i)]
\item $\E\left( K(F(X),F(Y))^2(1-G_0(X))(1-G_1(Y))\right)<\infty$,
\item $\E\left( K(F(X),F(X'))^2(1-G_0(X))(1-G_0(X'))\right)<\infty$,
\item $\E\left(K(F(X),F(X))(1-G_0(X))\right)<\infty$,
\item $\E\left( K(F(Y),F({Y'}))^2(1-G_1(Y))(1-G_1(Y'))\right)<\infty$, and
\item $\E\left(K(F(Y),F(Y))(1-G_1(Y))\right)<\infty$.
\end{enumerate}
\end{condition}
In the previous conditions, recall that $F=\eta F_0+(1-\eta)F_1$  is the pooled distribution, and  $F=F_0$  under the null hypothesis. In addition, we need a technical condition to deal with the randomness generated by  $\widehat F(x-)$, which is an analogue of the conditions needed in Theorem 4.2.1 of \citet{gill1980censoring}.

\begin{condition}\label{con:ExtraConditions}
Assume that for every $\varepsilon>0$, it holds that

\begin{enumerate}[i)]
    \item \label{con:ExtraConditions3}$ \lim_{t\to \tau} \limsup_{n\to \infty} \Prob\left( \int_{t}^{\tau} K(\widehat F(x-),\widehat F(x-))(1-G_0(x))dF_0(x)> \varepsilon\right)=0$, and
    \item \label{con:ExtraConditions4}$ \lim_{t\to \tau} \limsup_{n\to \infty} \Prob\left( \int_{t}^{\tau} K(\widehat F(x-),\widehat F(x-))(1-G_1(x))dF_1(x)> \varepsilon\right)=0$.
\end{enumerate}
\end{condition}
Note that our conditions are not easy to verify in practice as they require knowledge of the distributions involved,
 however, the conditions are trivially satisfied for continuous and bounded kernels on $[0,1)^2$, which is a property that most standard kernels, such as the Gaussian  kernel, satisfy.

Finally, in order to ease notation, we present our results in terms of the functions $\psi: \R_+ \to \R$ and $\psi^*: \R_+ \to \R$, defined by
\begin{align}
\psi(x)&=\frac{(1-G_0(x))(1-G_1(x))}{\eta(1-G_0(x))+(1-\eta)(1-G_1(x))},\label{Definition:psi}
\end{align}
and
\begin{align}
\psi^*(x)&=\frac{(1-G_0(x))(1-G_1(x))}{\eta(1-H_0(x))+(1-\eta)(1-H_1(x))}.\label{eqn:psistar}
\end{align}

\subsection{Limit distribution under the null hypothesis}
Our first result establishes that, under the null hypothesis, $(n_0n_1/n)Z_n$ converges in distribution to a limit random variable $Z$, when the number of data points tends to infinity.
\begin{theorem}[Limit distribution under the null]\label{thm:FinalDistr}
Assume Conditions \ref{con:conditions} and \ref{con:ExtraConditions} hold. Then, under the null hypothesis, we have that
\begin{eqnarray}
\frac{n_0n_1}{n}Z_n&\overset{\mathcal{D}}{\to}&Z = \int_0^{\tau}K(F_0(x),F_0(x))\psi(x)dF_0(x)+\frac{1}{\eta(1-\eta)}Y,
\end{eqnarray}
as $n$ approaches infinity, where
$Y=\sum_{i=1}^\infty\lambda_i(\xi_i^2-1)$, $\xi_1,\xi_2,\ldots$ are a collection of (potentially infinity) i.i.d. standard normal random variables, and $\lambda_1,\lambda_2,\ldots$ are non-negative constants. Additionally, the mean and variance of $Z$ are respectively given by
\begin{align*}
\E\left(Z\right)=\int_0^{\tau}K(F_0(x),F_0(x))\psi(x)dF_0(x)\quad\text{and}\quad&\Var\left(Z\right)=2\int_0^{\tau}\int_0^{\tau}K(F_0(x),F_0(y))^2\psi(x)\psi(y)dF_0(x)dF_0(y).
\end{align*}

\end{theorem}

In the previous theorem, the constants $\lambda_i$ are  the eigenvalues of an integral operator in the $L^2$ space associated with the measure induced by the triple $(T_i, \Delta_i, c_i)$ under the random covariates model. More details about the limit distribution are shown in Section \ref{sec:proofMain} along with the proof of Theorem~\ref{thm:FinalDistr}.

\subsection{Power under alternatives}

We continue by studying the asymptotic behaviour of our test statistic under the alternative hypothesis, that is, under the assumption that $H_1: F_0 \neq F_1$ holds. To this end, we first prove that $Z_n$ (without considering any scaling) converges to a deterministic value. 

Define the measures $\nu_0$ and $\nu_1$ on $\R_+$ by
\begin{align}
\nu_0(A)=\int_{A}\psi^*(s)S_1(s)dF_0(s)\quad&\text{and}\quad
\nu_1(A)=\int_{A}\psi^*(s)S_0(s)dF_1(s),\label{eqn:measuresInf}
\end{align}
where $\psi^*$ is defined in Equation \eqref{eqn:psistar}, and define their embeddings into $\mathcal{H}$ by
\begin{align}\label{eqn:fancyEmbedding1}
\phi_{0}(\cdot) = \int_0^\tau K(F(y),\cdot)d\nu_0(y)\quad\text{and}\quad\phi_{1}(\cdot) = \int_0^\tau K(F(y),\cdot)d\nu_1(y),
\end{align}
respectively. The measures $\nu_0$ and $\nu_1$ can be understood as the  population measures, respectively, of the empirical measures $\nu_0^n$ and $\nu_1^n$ defined in Equation \eqref{eqn:measuresn}. We will show that, under appropriate conditions, the embeddings $\phi_0^n$ and $\phi_1^n$, defined in Equation \eqref{eqn:Embeddingsn}, converge to $\phi_0$ and $\phi_1$, respectively, where convergence is with respect to the norm of $\mathcal H$. This fact together with Theorem~\ref{thm:ZndoubleIntegral} yields the following result:

\begin{theorem}\label{thm:convergenceAlternative}
Assume Conditions \ref{con:conditions} and \ref{con:ExtraConditions} hold, then
\begin{eqnarray}
Z_n\overset{\Prob}{\to}\|\phi_0-\phi_1\|^2_{\mathcal{H}}.
\end{eqnarray}
\end{theorem}

At this point, it should be clear that our test is consistent under the alternative hypothesis if  $\phi_0 \neq \phi_1$ whenever $\nu_0 \neq \nu_1$.  We will show that  if $F_0 \neq F_1$ then  $\nu_0\neq \nu_1$, however, this result does not immediately extrapolate to  $\phi_0$ and $\phi_1$. A sufficient condition to ensure that $\phi_0$ and $\phi_1$ are different when $F_0\neq F_1$  is that the kernel $K$ is c-universal (see definition in Section~\ref{sec:RHKSintro}).

\begin{corollary}[Consistency]\label{thm:obnibus}
Suppose that Conditions \ref{con:conditions} and \ref{con:ExtraConditions} hold, and that the kernel $K$ is c-universal. Additionally, assume that $1-G_c(x)=0$ implies $S_c(x) = 0$, for any $c \in \{0,1\}$. Then, under the alternative hypothesis, we have that $n_0n_1/nZ_n\to\infty$, which implies that 
\begin{align*}
\Pr\left(\frac{n_0 n_1}{n}Z_n > Q_n(1-\alpha)\right) \to 1,
\end{align*} 
where $Q_n(1-\alpha)$ denotes the $(1-\alpha)$-quantile of the random variable $(n_0n_1/n)Z_n$ under the null hypothesis.
\end{corollary}

We remark that the previous result uses the $(1-\alpha)$-quantile of $(n_0n_1/n)Z_n$ under the null hypothesis, which  converges to a finite quantity due to Theorem~\ref{thm:FinalDistr}.  

Notice that Corollary~\ref{thm:obnibus} establishes consistency for all the alternatives under the assumption of a $c$-universal kernel $K$,  nevertheless, even if the kernel $K$ is not $c$-universal, we  can still ensure consistency for particular alternatives, as we show in the following proposition.
\begin{proposition}
Assume that Conditions \ref{con:conditions} and \ref{con:ExtraConditions} hold, and suppose that $d\Lambda_1(x) = e^{\theta \bar \omega(F_0(x))}d\Lambda_0(x)$ for some $\bar \omega \in \mathcal H$ and $\theta \neq 0$. Then $(n_0n_1/n)Z_n \to \infty$, and thus the test is consistent for such an alternative.
\end{proposition} 
Note that if $d\Lambda_1(x) = e^{\theta \bar \omega(F_0(x))}d\Lambda_0(x)$ for some $\bar \omega \in \mathcal H$ and $\theta \neq 0$, then $\liminf_{n\to \infty}\LRT_n(\bar \omega)^2>0$. In this case, 
\begin{align*}
\liminf_{n \to \infty} Z_n = \liminf_{n \to \infty}\sup_{\omega \in \mathcal H, \|\omega\|_{\mathcal{H}} \leq 1} \LRT_n(\omega)^2 \geq \liminf_{n \to \infty}\LRT_n(\bar \omega)^2>0.
\end{align*}
Therefore $(n_0n_1/n)Z_n \to \infty$, and thus the test is consistent for such a particular alternative. In general, any argument that ensures consistency of $\LRT_n(\omega)^2$ for some $\bar \omega$ in $\mathcal H$ also applies to our test statistic $Z_n$, as long as the limit distribution of $(n_0n_1/n)Z_n$ exists under the null hypothesis.

\section{Wild Bootstrap Implementation}\label{sec:Wild}

Given $\alpha \in (0,1)$, we construct a statistical test of level $\alpha$ by rejecting the null hypothesis if $(n_0 n_1/n)Z_n > Q_n(1-\alpha)$, where $Q_n(1-\alpha)$ is the $(1-\alpha)$-quantile of the distribution of $(n_0 n_1/n)Z_n$ under the null hypothesis. As the distribution of $(n_0 n_1/n)Z_n$ is usually unknown, we can use its asymptotic distribution, given in Theorem \ref{thm:FinalDistr}, to approximate the rejection region. Unfortunately, excluding exceptional cases, the limiting distribution is rather complex, and thus computing the asymptotic $(1-\alpha)$- quantile is very hard. In order to carry out our test, we introduce a simple Wild Bootstrap implementation of our testing procedure.

Recall, from Equation~\eqref{eqn:ZnComputationImplementation}, that our test statistic $Z_n$ can be written as follows:
\begin{align}
Z_n&=\left(\frac{n}{n_0n_1}\right)^2\sum_{i=1}^n \sum_{j=1}^n K(\widehat F(X_i-), \widehat F(X_j-)) L(X_i)L(X_j) (-1)^{c_i+c_j}\frac{\Delta_i}{Y_{c_i}(X_i)}\frac{\Delta_j}{Y_{c_j}(X_j)}.
\end{align}
Notice that the expression given for $Z_n$ looks like a V-statistic. However, what prevents $Z_n$ from being a V-statistic is that it depends on the functions $\widehat{F}(x-)$, $Y_{0}(x)/n_0$ and $Y_1(x)/n_1$, and $n/(n_0n_1)L(x)$, which are all random functions that depend on all data points. In Appendix~\ref{sec:nicerRep} and Appendix \ref{sec:equivalence}, we prove that we can replace these functions by their respective limits plus some small additive error term that vanishes sufficiently fast as the number of observations grows. From this result, we deduce that $Z_n$ can be approximated by a V-statistic, suggesting that the standard  Wild Bootstrap sampling scheme of \citet{dehling94random} can be used  to approximate the asymptotic distribution of  $n/(n_0n_1) Z_n$.  

Let $\mathcal W = (W_i)_{i\in[n]}$ be a sequence of i.i.d. random variables satisfying that $\E(W_i) = 0$ and $\E( W_i^2) = 1$, and such that they are independent of any other source of randomness, in particular, from the observations. Then, we define the Wild Bootstrap statistic associated with $Z_n$ by
\begin{align}
Z_n^{\mathcal W}&=\left(\frac{n}{n_0n_1}\right)^2\sum_{i=1}^n\sum_{j=1}^n W_i W_j K(\widehat F(X_i-), \widehat F(X_j-)) L(X_i)L(X_j)(-1)^{c_i+c_j} \frac{\Delta_i}{Y_{c_i}(X_i)}\frac{\Delta_j}{Y_{c_j}(X_j)}.\label{eqn:wildBoostrapDefi}
\end{align}
Similarly to the expression given in Equation~\eqref{eqn:Z_nMatrix}, we can write  $Z_n^{\mathcal W}$ as
\begin{align}\label{eqn:Z_nWMatrix}
Z_n^{\mathcal W}= \left(\frac{n}{n_0n_1}\right)^2(\mathbf{W}{\mathbf{V}})^{\intercal} \widehat{ \mathbf K}({\mathbf{W}\mathbf{V}}),
\end{align}
where $\mathbf{W}\in\R^n\times\R^n$ is the diagonal matrix whose entries are given by $\mathbf{W_{jj}} = W_j$. Also, recall that $\mathbf{V} \in \R^n$ is defined by $\mathbf{V}_j = L(X_j)(-1)^{c_j}\Delta_j/Y_{c_j}(X_j)$  and  $\widehat{\mathbf{K}}\in \R^n\times \R^n$ is defined by $\widehat{\mathbf{K}}_{ij} = K(\widehat F(X_i),\widehat F(X_j))$. This matrix expression is very convenient for the computational implementation of our method.

The following theorem states the correctness of our Wild Bootstrap approach:

\begin{theorem}\label{thm:wildbootstrap}
Assume  that Conditions~\ref{con:conditions} and \ref{con:ExtraConditions} are satisfied. Then, under the null hypothesis, it holds that for any $x \in \R_+$,
\begin{align*}
\Prob\left(\frac{n_0n_1}{n}Z_n^{\mathcal W} \geq x \bigg| (X_i, \Delta_i, c_i)_{i\in[n]} \right) \overset{\Prob}{\to}  \Prob( Z \geq x),
\end{align*}
where $Z$ is the random variable defined in Theorem~\ref{thm:FinalDistr}.
\end{theorem}

With the previous ingredients we are ready to describe our testing procedure, which relies on approximating the $(1-\alpha)$-quantile of the distribution of $Z_n$ by the quantiles of $Z_n^{\mathcal W}$. Since we can freely sample independent copies  from $Z_n^{\mathcal W}$ given the data points, we can estimate the quantile by Monte Carlo simulations. Our algorithm is as follows:

\begin{enumerate}[i)]
\item Set level $\alpha\in (0,1)$ of the test and let $N$ be a large integer,
\item  Sample $N$ independent copies of the Wild Bootstrap statistic using Equation~\eqref{eqn:Z_nWMatrix},
\item  Compute the $(1-\alpha)$-quantile of the previous sample and call it $Q_n^{\mathcal{W}}(1-\alpha)$,
\item  Compute the test statistic $Z_n$ using Equation~\eqref{eqn:Z_nMatrix},
\item  Reject the null-hypothesis if $Z_n>Q_n^{\mathcal{W}}(1-\alpha)$, otherwise do not reject it.
\end{enumerate}

\section{Simulations}\label{sec:Simu}
We perform an empirical evaluation of our methods in which the ground truth is known. To this end, we consider two different settings: a \emph{proportional hazard functions} setting (in which the classic log-rank test is provably the most powerful), and a \emph{time-dependent hazard functions} setting, including Weibull and periodic hazard functions.  All our experiments consider the same null cumulative hazard function $\Lambda_0(t)=t$, corresponding to the cumulative hazard function of an exponential random variable with mean 1. We choose $\Lambda_1(t)$ belonging to one of the following parametric families:
\begin{enumerate}[i)]
\item \textbf{Proportional hazards}: for this case we consider $\Lambda_1(t)$ belonging to the parametric family given by $\Lambda(t;\theta) = \theta t$ with $\theta\in[0,2]$. Observe that $\Lambda_0(t)$ is recovered when $\theta = 1$.
\item \textbf{Weibull (polynomial) hazards}: we consider $\Lambda_1(t)$ belonging to a family of Weibull cumulative hazard functions $\Lambda(t;\theta) =  t^{\theta}$ with $\theta \in [0,2]$. Notice that $\theta = 1$ recovers the null. 
\item \textbf{Periodic hazards}: we consider $\Lambda_1(t)$ belonging to the family  given by  $\Lambda(t;\theta) = t-\sin(\pi \theta t)/(\pi\theta)$ with $\theta\in(0,15]$. Notice that  $\lim_{\theta\to\infty}\Lambda(t;\theta)=\Lambda_0(t)=t$, recovering the null hypothesis.
\end{enumerate}
Figure \ref{Figure:Chf} shows the behaviour of the cumulative hazard functions of the parametric families previously described, for different values of $\theta$.

\begin{figure}
\begin{center}
\includegraphics[width=0.95\linewidth]{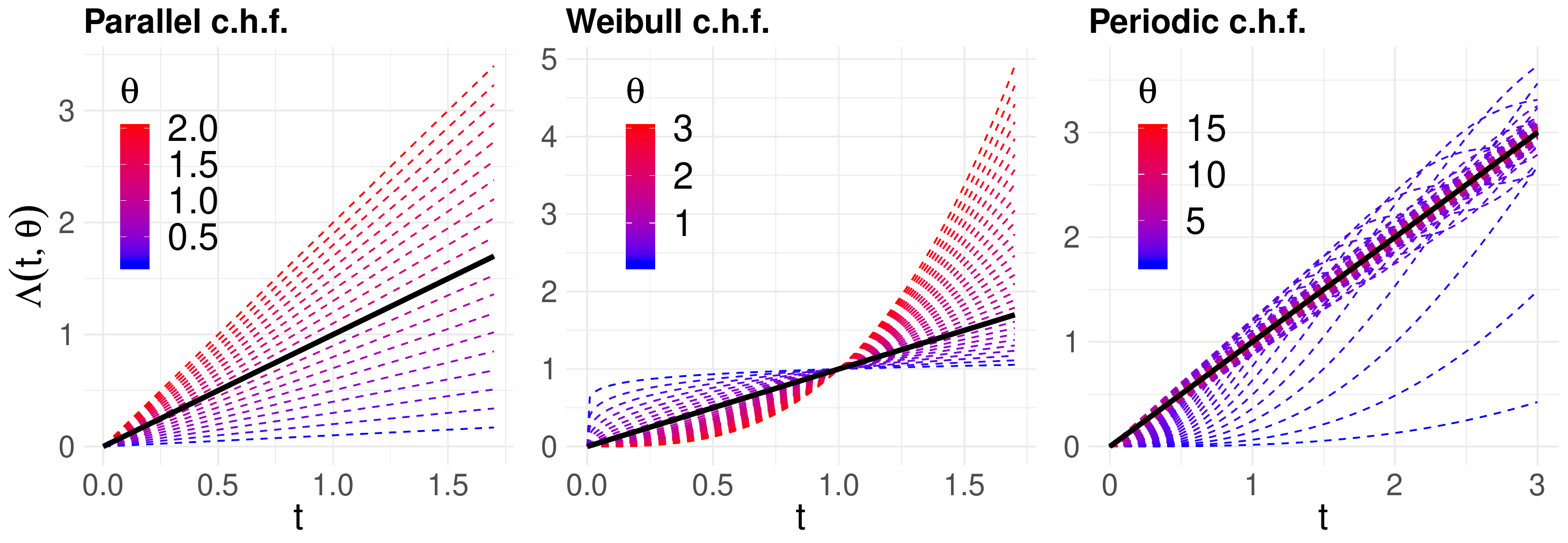}
\end{center}
\caption{Set of cumulative hazard functions for the proportional, Weibull and periodic experiments. Solid black line denotes the null cumulative hazard $\Lambda_0(t)=t$. Observe that each family recovers the null cumulative hazard function $\Lambda_0(t)$ for some specific value $\theta_0$. }\label{Figure:Chf}
\end{figure}

\subsection{Implementation}

\subsubsection{Kernels}\label{Sec:Kernels}
The heart of our testing-approach is undoubtedly the kernel function. In Section~\ref{sec:Recovering} we showed that specific choices of the kernel function lead to some existing tests.  In our experiments we use some of these choices as well as a kernel that  is $c$-universal. We define kernels in the following categories:
\begin{enumerate}
\item \textbf{Log-rank kernels (LRP and LRC)}: In Section~\ref{sec:weilogranktest}, we considered kernels of the form $K(x,y)=\omega(x)\omega(y)$, recovering the well-known weighted log-rank tests. In our experiment we choose $\omega$  to be equal to i) $\omega_1(x) = 1$ and ii)  $\omega_2(x) = (x-1/2)$. For i) we recover the classical log-rank test (LRP), used to test proportional hazard functions, while for ii) we recover a weighted log-rank test (LRC), designed to discover if the hazard functions cross around the median of the distribution $F$. 

\item \textbf{Projection kernels (P2W and P4W)}: we follow the approach described in Section~\ref{sec:projectiontest}, which recovers the testing procedure of \citet{Brendel2013weighted}. In particular, we choose kernel functions based on the subspace generated by the weight functions i) $\{1,x\}$ and ii) $\{1,x,x^2, x^3\}$. We denote the kernels in i) and ii) by P2W and P4W, respectively, making a clear reference to the dimension of the subspace.

\item \textbf{Pearson-type kernels (Per4 and Per5)}: we consider the Pearson-type kernel functions defined in Section~\ref{sec:pearsontest}, particularly in Equation~\eqref{eqn:pearsonTest}. For our experiments, we partition the space into 4 and 5 disjoint regions, and denote the kernels by Per4 and Per5, respectively.

\item \textbf{Squared exponential kernel (SEK)}: we consider the squared exponential kernel (SEK) defined by $K(x,y) = \exp\{-(x-y)^2/\sigma^2\}$ with $\sigma=0.1$. This kernel is $c$-universal, hence by Corollary \ref{thm:obnibus}, it leads to an omnibus test. Better results may be obtained by optimising the parameter $\sigma$ for the problem at hand. A well-known heuristic is to choose $\sigma$ as the median of the pairwise differences of the observations \citet{Scholkopf2001}. \end{enumerate}.

\subsubsection{Computer Implementation}
Our experiments are implemented in \texttt{R} following the Wild Bootstrap approach described in Section \ref{sec:Wild}, choosing Rademacher random variables for the Wild Bootstrap weights $\mathcal{W}$. We compute our test statistic $Z_n$ and the Wild Bootstrap statistic $Z_n^{\mathcal{W}}$ by using the quadratic form expressions given in Equations~\eqref{eqn:Z_nMatrix} and~\eqref{eqn:Z_nWMatrix}, respectively. By using these expressions, it is possible to give a simple and fast implementation of our testing procedure, indeed, 1000 repetitions of our testing procedure (using 1000 Wild Bootstrap samples) takes just a couple of minutes in a standard commercial laptop, for $n_0=n_1 = 100$. For each experiment, we consider sample sizes of $30$ and $100$ observations per group, and we choose a censoring distribution generating $10\%$ and $30\%$ of censored observations.

\subsection{Type-I error}
In our first experiment, we verify that our tests achieve a correct Type-I error of $\alpha=0.05$ for each combination of sample size and censoring percentage,  for  the kernels described in Section~\ref{Sec:Kernels}. For each different combination of parameters, we run our test in 1000 simulated datasets. Table \ref{table:significance5} shows the results. In general, the Wild Bootstrap approach has no problem reaching the correct level for even censoring percentages. Uneven censoring causes a few problems, so the user should be careful when applying this method in the latter setting. Arguably, the squared exponential  kernel (SEK) is the most robust. It is worth recalling that projection tests use a random kernel (it depends on the data points) which might impact its performance. Additionally, while not reported here, we tested other significance levels, obtaining similar results.

\begin{table}
\centering
\begin{tabular}{llll|rr|rr|rr|r|}

\cline{5-11}
                           &                            &                          &       & \multicolumn{2}{c|}{Log-Rank}                        & \multicolumn{2}{c|}{Projection}                          & \multicolumn{2}{c|}{Pearson-type}                    & \multicolumn{1}{c|}{SEK}      \\ \hline
\multicolumn{1}{|l}{$n_0$} & \multicolumn{1}{l|}{$n_1$} & $\text{cen}_0$                    & $\text{cen}_1$ & \multicolumn{1}{l}{LRP} & \multicolumn{1}{l|}{LRC} & \multicolumn{1}{l}{P2W} & \multicolumn{1}{l|}{P4W} & \multicolumn{1}{l}{Per4} & \multicolumn{1}{l|}{Per5} & \multicolumn{1}{l|}{SEK} \\ \hline
\multicolumn{1}{|l}{30}    & \multicolumn{1}{l|}{30}    & 10\%                     &10\%  &4.5                      &5                         & 5.1                       &4.7                          & \textbf{6.1}                      &\textbf{5.7}                       &5                             \\
\multicolumn{1}{|l}{}      & \multicolumn{1}{l|}{}      & 10\%                     &30\%  &5.4                      &4.4                       &4.5                       &4.1                          & 5                        &5.1                       &4.2                           \\
\multicolumn{1}{|l}{}      & \multicolumn{1}{l|}{}      & 30\%                     &30\%  &4.3                      & 4.7                       &4.2                       & 4.4                          & 4.7                      & 4.4                       & 4.2                           \\ \hline
\multicolumn{1}{|l}{30}    & \multicolumn{1}{l|}{100}   & 10\%                     & 10\%  & 3.9                      &\textbf{ 5.7}                      & 5                         & 5.1                          & 5.2                      & 4.9                       & 5.2                           \\
\multicolumn{1}{|l}{}      & \multicolumn{1}{l|}{}      & 10\%                     & 30\%  & \textbf{6.7}                      & 4.6                       & \textbf{7.1} & 5.1                          & 5.2                      & 4.5                       & 4.7                           \\
\multicolumn{1}{|l}{}      & \multicolumn{1}{l|}{}      & 30\%                     & 10\%  &5                          &\textbf{7.3}                        &\textbf{6.1}                           &4.9                              &\textbf{5.9} &\textbf{5.6}                         &5                               \\
\multicolumn{1}{|l}{}      & \multicolumn{1}{l|}{}      & 30\%                     & 30\%  & 3.9                      & 5.1                       & 3.8                       & 3.9                          & 4.4                      & 4.8                       & 3.8                           \\ \hline
\multicolumn{1}{|l}{100}   & \multicolumn{1}{l|}{100}   & 10\%                     & 10\%  & 4.7                      & 4.2                       & 5.2                       & \textbf{5.9}                          & 4.9                      & 4.7                       & 4.3                           \\
\multicolumn{1}{|l}{}      & \multicolumn{1}{l|}{}      & 10\%                     & 30\%  & 5                        & \textbf{5.9}                      & 5.4                       & \textbf{6.6}                         & \textbf{6.1} & \textbf{5.8} &\textbf{6.1}                          \\
\multicolumn{1}{|l}{}      & \multicolumn{1}{l|}{}      & \multicolumn{1}{r}{30\%} & 30\%  & 5.1                      & 5.3                       & 5.5                       & 4.7                          & 5.1                      & 4.8                       & 4.8                           \\ \hline
\end{tabular}
\caption{\label{table:significance5}Significance Level at $\alpha = 5\%$. $n_c$ stands for the number of data points of group $c$, and $\text{cen}_c$ indicates the percentage of censored data. }
\end{table}

\subsection{Power Simulations}
We provide an empirical evaluation of our testing procedure for each of the settings described in Section~\ref{sec:Simu}: proportional, Weibull and periodic hazard functions.  As previously mentioned, the null cumulative hazard function is given by $\Lambda_0(t)=t$. The power is estimated by repeating our testing procedure over 1000 simulated datasets for each combination of sample size and censoring percentage. Results for the proportional, Weibull and periodic settings are shown in Figures \ref{fig:prophazads}, \ref{fig:polyhazards}, and \ref{fig:periodichazads}, respectively.
We give a few comments and remarks about our experiments.
\begin{enumerate}
\item We report a very small fraction of our experiments as all the results are qualitatively the same as the ones shown in Figures \ref{fig:prophazads}, \ref{fig:polyhazards}, and \ref{fig:periodichazads}.

\item We only report results for the  kernel Per5 since the kernel Per4 has an almost identical behaviour.

\item The LRP test is equivalent to the score test for the proportional  hazard functions model; thus, it can be deduced that it is the most powerful test for local alternatives under this model. The previous statement is supported by Figure \ref{fig:prophazads}, where we observe an excellent performance of the LRP test in the setting of proportional hazard functions alternatives. Observe that the LRP test loses all of its power (nearly zero power) for the Weibull and periodic hazard alternatives as shown in Figures \ref{fig:polyhazards}, and \ref{fig:periodichazads}.

\item The LRC test is equivalent to the weighted log-rank test with weight $w(x) = x-1/2$. It can be observed in Figure \ref{fig:polyhazards}, that the LRC test has a very good performance in the setting of Weibull hazard alternatives, but its power is relatively low  in the other two settings, as shown in Figures \ref{fig:prophazads} and \ref{fig:periodichazads}. An explanation for this behaviour follows from the fact that the LRC kernel is designed to be optimal at detecting a cross around the median of the pooled distribution $F$, and thus, it will not be a good fit for the proportional hazards experiment. Also, for more complex hazard functions, such as those described in the periodic setting, we can observe more than one cross occurring, which explains the poor performance of the LRC test.

\item We  observe different behaviours for the projection kernels P2W and P4W. In the proportional hazard functions setting, the P2W test has the best performance after the LRP test, which is explained by the fact that this kernel is constructed considering the subspace generated by $\{1,x\}$, where the weight $\omega(x)=1$ is known to be optimal for proportional hazard alternatives. While P4W also includes this weight (it is generated by $\{1,x,x^2,x^3\}$), the fact that it considers a larger space of possible alternatives  makes the test more data-expensive resulting in a loss of power. For the Weibull hazard functions setting, we observe that both projection tests, P2W and P4W, have an overall good performance. This behaviour can be explained by the fact  that both tests consider projections on polynomials, and the Weibull hazard functions are, indeed, polynomials. In the periodic case, both kernels have a substandard behaviour due to the fact that the hazard structure is rather different to a polynomial of finite degree. Note that with more data it seems that the tests do not improve when compared to the best kernel, in this case, the squared exponential kernel SEK.

\item The Pearson-type kernel Per5 has a consistent behaviour, being neither too good nor too bad. Disadvantages are that the user needs to specify beforehand a partition of the space.

\item The squared exponential kernel, SEK, gives overall good results. In one hand, while in the presence of more `structured' data, i.e., proportional or Weibull (polynomial) hazard functions, simpler kernels give better results, the SEK still has a good performance. On the other hand, in presence of  more `complex' data, its performance is better than other kernels. In general, the SEK performs better than Pearson-type kernels, suggesting that the SEK is a better choice for an all-around kernel. 

\item In general, it seems that for nice structured data, simpler kernels (leading to simpler methods)  have better results. On the other hand, for complex data, a more complex kernel seems to be a good option. Overall, we think that the SEK is the best option as it has a robust behaviour in simple data-scenarios, and it outperforms other tests in more complex scenarios. Also, as shown in Table \ref{table:significance5}, this kernel has close to no problems reaching the correct Type-I error.
\end{enumerate}

\begin{figure}
   \includegraphics[width=\textwidth]{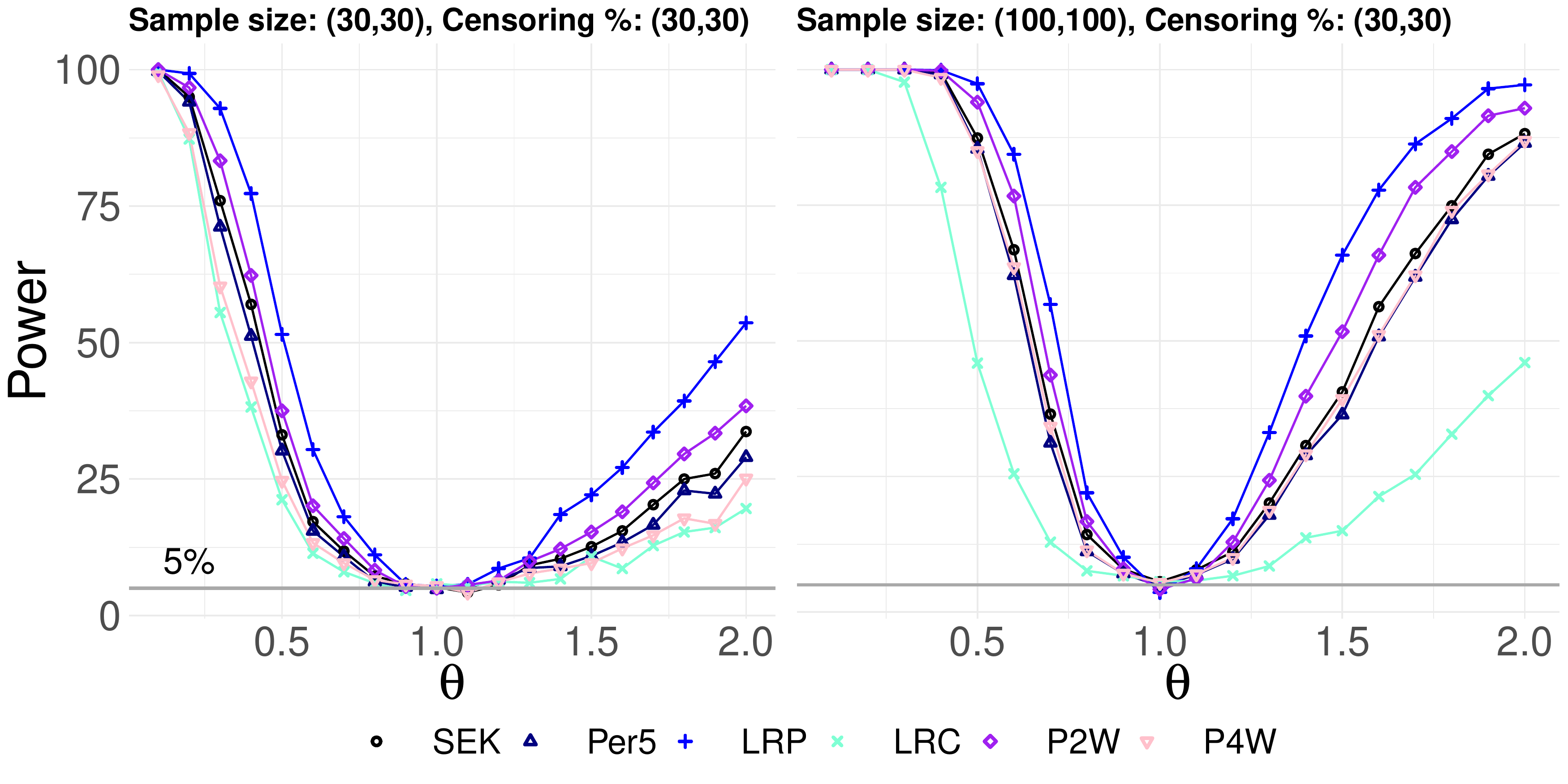}
  \caption{Proportional hazards experiment. Left: group sizes 30 and censoring percentage $30\%$. Right group sizes 100 and censoring percentage $30\%$.}
  \label{fig:prophazads}
\end{figure}

\begin{figure}
   \includegraphics[width=\textwidth]{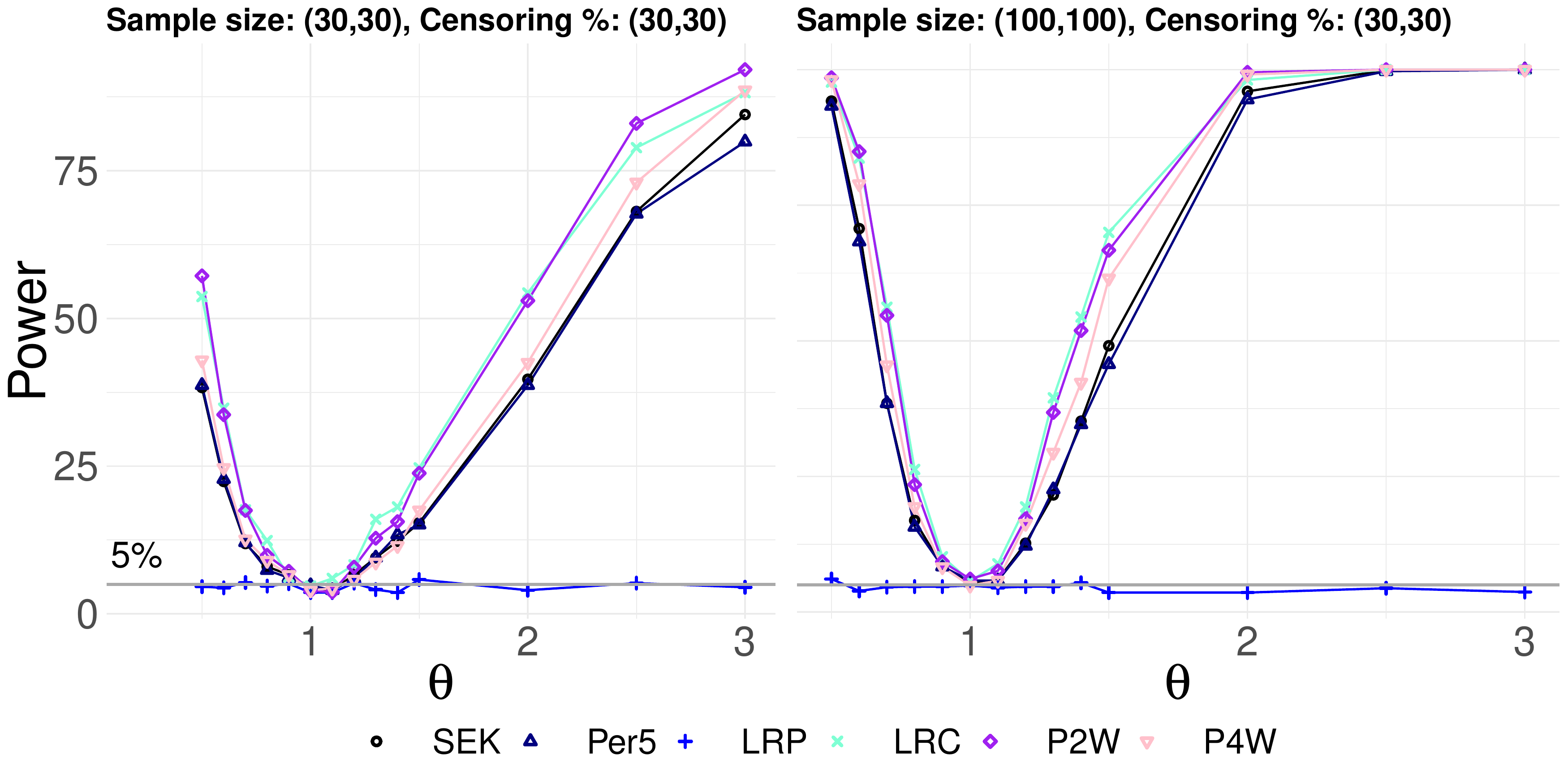}
  \caption{Polynomial (Weibull) hazards experiment. Left: group sizes 30 and censoring percentage $30\%$. Right group sizes 100 and censoring percentage $30\%$.}
	\label{fig:polyhazards}
\end{figure}

\begin{figure}
   \includegraphics[width=\textwidth]{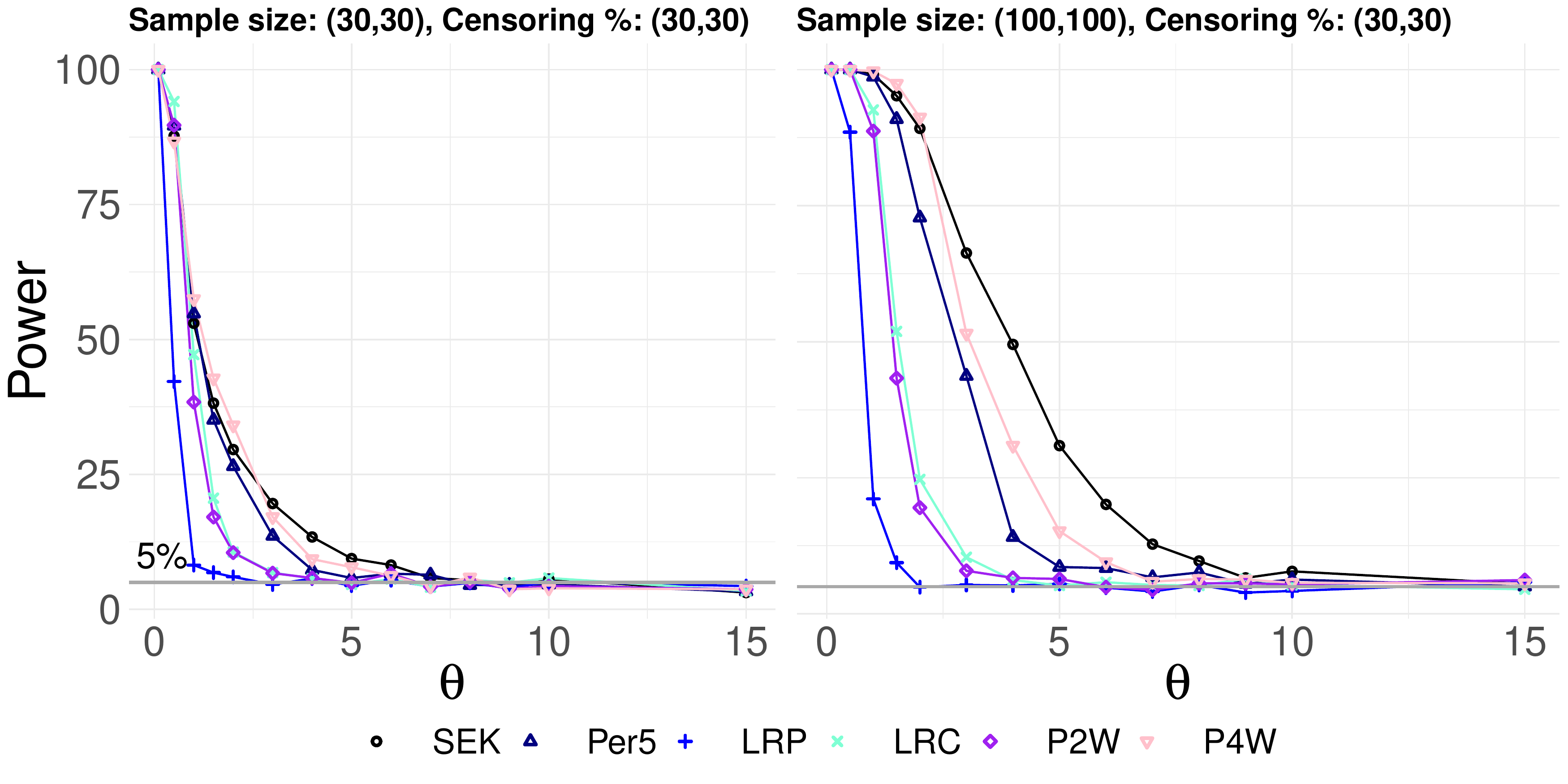}
  \caption{Periodic hazards experiment. Left: group sizes 30 and censoring percentage $30\%$. Right group sizes 100 and censoring percentage $30\%$.}
  \label{fig:periodichazads}
\end{figure}

\section{Real data}\label{sec:Real}
We consider the Gastrointestinal Tumor Study Group data (GTSG), \citet{stablein1981analysis}, available in the \texttt{`coin' R-package}. The data considers a randomised clinical trial in the treatment of locally advanced, non-resectable gastric carcinoma. In this study, 42 patients were treated by using chemotherapy alone, while 45 patients were treated by using a combination of chemotherapy and radiation therapy. The aim of the study is to detect differences between these treatments. Kaplan-Meier curves for each group are shown in Figure \ref{Figure:Real}. The null hypothesis is that there is no difference between the treatments. We apply our test considering the 7 different kernels described in Section \ref{Sec:Kernels}. The corresponding p-values (approximated by using our Wild Bootstrap approach) are shown in Table \ref{Table:Real data}.

\begin{figure}
\begin{center}
\includegraphics[width=1\linewidth]{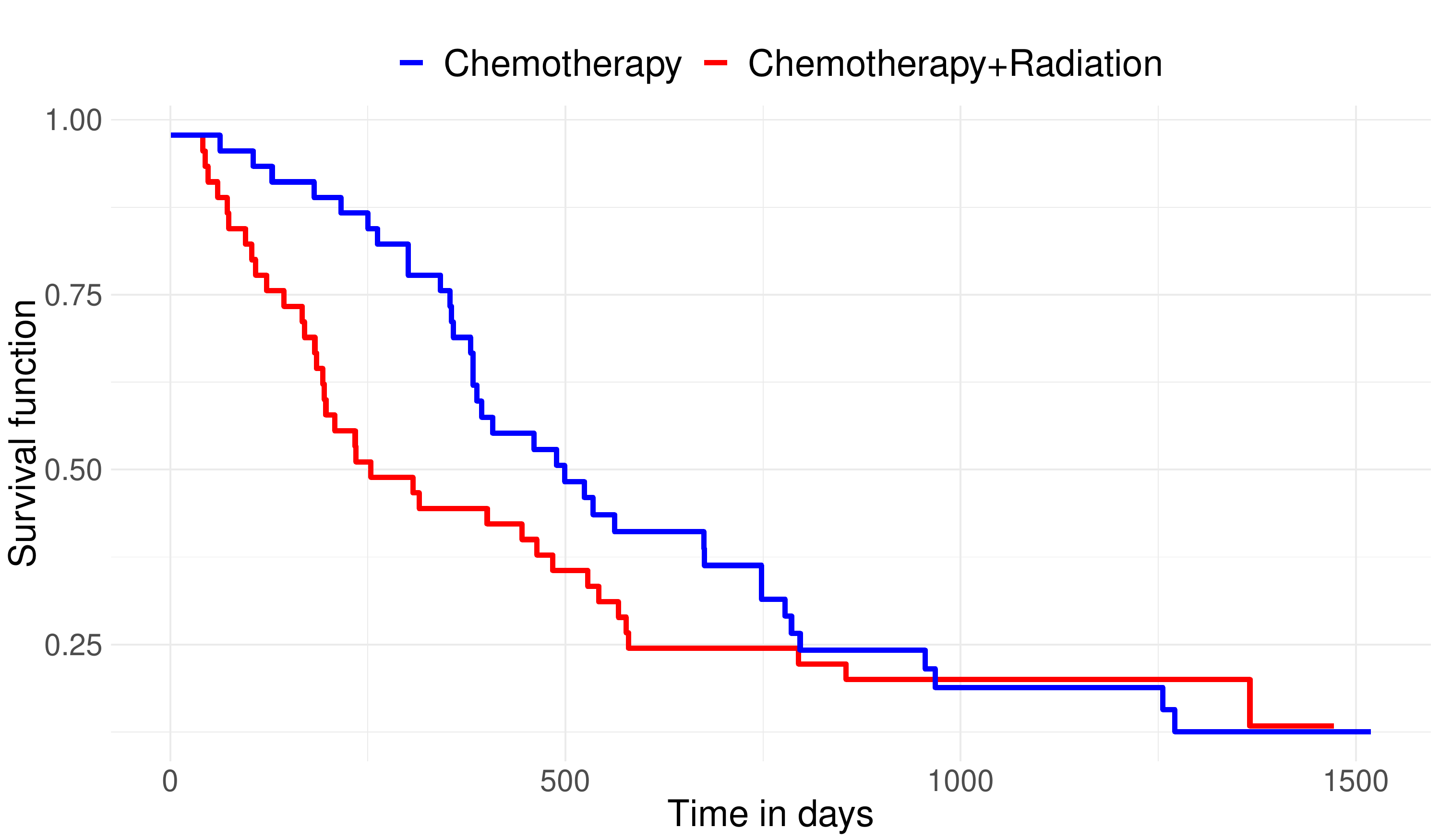}
\end{center}
\caption{Kaplan-Meier curves for the patients treated by using chemotherapy alone versus a combination of chemotherapy and radiation therapy  }\label{Figure:Real}
\end{figure}
% latex table generated in R 3.4.4 by xtable 1.8-3 package
% Wed Jan 30 16:24:04 2019

\begin{table}

\begin{center}
\begin{tabular}{rrrrrrrrr}
  \hline
  &SEK & Per4&Per5 & LRP & LRC&P2W&P4W \\ 
  \hline
\text{p-values} & 0.0053 &0.0151& 0.0222& 0.2531 & 0.0011 &0.0051&0.0228\\ 
   \hline
\end{tabular}
\caption{\label{Table:Real data}$p$-values obtained (using the Wild Bootstrap approach) for the comparison of the treatments `Chemotherapy' and `Chemotherapy+Radiation' for the GTSG data.}
\end{center}
\end{table}

All the tests, except for the classical log-rank test, reject the null hypothesis at  $5\%$ level of significance. This outcome is quite reasonable as the survival functions cross, as shown in Figure \ref{Figure:Real}. Indeed, the smallest $p$-values are given by the SEK, LRC and the P2W  tests which reject the null hypothesis at  $1\%$ level of significance. This is not a surprising behaviour of LRC and P2W, as these  kernels are tailored to detect crossings. The SEK  also performs very well which is quite satisfying as this kernel is not particularly designed for the setting of crossing hazards.

\section{Conclusion}
We have introduced an RKHS-based testing procedure for the standard problem of two-sample hypothesis testing in the framework of right censored data. Our test statistic is the supremum of weighted log-rank statistics, where the weights belong to the unit ball of an RHKS. While our test statistic is apparently very complex, its evaluation becomes analytically tractable due to the reproducing property of reproducing kernel Hilbert spaces. Indeed, our test statistic can be written as a quadratic form as in Equation~\eqref{eqn:Z_nMatrix}, and it can be fully characterised by a kernel function $K$. This simple structure allows us to derive asymptotic properties of our test, and it suggests that the standard Wild bootstrap approach can be used to approximate the rejection region.
We also showed that our test can be seen as a natural infinitely-dimensional generalisation of other well-known test statistics based on log-rank statistics. Finally, we performed a simulation study which compares the results of the test statistic for different kernel functions.

To finish the paper we discuss some of our results and potential research ideas. First, as shown in Theorem~\eqref{thm:FinalDistr}, the asymptotic distribution of our test statistic under the null hypothesis is, in general,  intractable, and thus we need to rely on bootstrapping techniques to approximate the rejection region. While, for us, the natural option is to consider a Wild Bootstrap approach, it is also possible to consider a permutation approach as the one used for weighted log-rank tests \citep{Neuhaus1993Conditional}. In the case of equal censoring, the permutation approach has the advantage of being finitely exact, however, in the general case of censoring, the permutation approach is not directly applicable since the limit distribution  of our test statistic depends on the censoring distribution of each group. For weighted log-rank tests, this problem is fixed by standardising the test statistic leading to an asymptotically distribution-free test. Unfortunately, it is not clear how to standardise (if possible) our test statistic, and we think this is a non-trivial task, especially for infinitely-dimensional kernels, hence, we leave this problem as future research. 

A natural question to ask is: How do we choose a kernel function? Unfortunately, nobody can provide an answer for such a question yet and choosing kernel functions is an active field of research in Statistics and Machine Learning, as it is an issue that happens in several other contexts such as Gaussian Processes inference,  support vector machine, kernel regression, etc. In practice, our simulation study suggests that simple (finite-dimensional) kernels, e.g., polynomials, are good enough if we know in advance that the hazard function has a relatively simple form, whereas complex kernels are better suited for complex hazard functions. In general, the squared exponential kernel is a very safe choice as it performs well in both settings. Nevertheless, we should not expect it to perform well in every setting, as \citet{janssen2000global} proved that for a finite number of data points,  any nonparametric test has preferences for a finite-dimensional subspace of alternatives, and outside this subspace, the power of the test is almost flat in balls of alternatives. This means that, for a fixed number of data points, it is possible to find alternatives for which the squared exponential kernel has fails, however, due to the complex form of the kernel, it may be difficult to construct such an alternative, and we do not expect that to happen in practical settings. Note that, by choosing a reproducing kernel, Theorem~\ref{thm:FinalDistr} shows that our test fixes most of its power on a finite number of directions, which are given by the eigenfunctions with larger eigenvalues in the spectral decomposition of the kernel, and thus, with a finite number of  data points we expect the test to concentrate its power around the first $m$  eigenfunctions with larger eigenvalues, where $m$ is some constant depending on $n$, and as long as $n$ grows, $m$ should be growing as well.

Finally, we give a few comments about our technical results. Our asymptotic analysis uses the fact that our test statistic is a double stochastic integral (Equation~\eqref{eqn:ExplicitZ}), and particularly, under the null hypothesis, those stochastic integrals are with respect to martingales. Our main tool to study these objects is Theorem~\ref{thm:newGeneralTheorem}, which allows us to control double stochastic integrals with respect to a special class of random integrands $h$, which includes our integrand $K(\widehat F(x-),\widehat F(y-))L(x)L(y)$, as well as several others. Our theorem is powerful enough to study other type of collection of log-rank statistics, for example, we can consider the test statistic 
\begin{align*}
    \sup_{\omega \in \mathcal H, \|\omega\|_{\mathcal H}^2\leq 1} \int_0^{\tau_n} \omega(Y(x)/n)L(x)(d\widehat \Lambda_0(x)-d\widehat \Lambda_1(x)),
\end{align*}
which is similar to our test statistic (c.f., Equations~\eqref{eqn:supStat} and~\eqref{eqn:log-rank}), but replacing $\widehat S(x-)$ by $Y(x)/n$ (See  \cite[Chapter 7]{Flemming91}). In general, replacing $\widehat S(x-)$ by any reasonable predictable process gives a test statistic that can be analysed by our methods, and furthermore, our tool can be applicable to an even larger class of processes. The analysis of other type of test statistics is subject of future research.

\section*{Correspondence}
Nicolas Rivera.\\
Address: Computer Laboratory, University of Cambridge.
William Gates Building, 15 JJ Thomson Ave, Cambridge CB3 0FD, UK.\\
e-mail: nr454 at cam dot ac dot uk

\bibliographystyle{plainnat}
\bibliography{ref}
\newpage

\section*{Appendix}
\appendix

\section{Preliminary Results}
\subsection{Some Results for Counting Processes}
\begin{proposition}\label{prop:supConverSH}
The following results hold:
\begin{enumerate}[i)]\item $\lim_{n\to \infty} \sup_{t\leq \tau} |\widehat S(t)-S(t)| = 0 $ a.s.,\label{prop:supConverSH1}
\item $\lim_{n \to \infty}\sup_{t \leq \tau} |Y(t)/n-H(t-)| = 0$  a.s., \label{prop:supConverSH2}
\end{enumerate}
\end{proposition}
Item \ref{prop:supConverSH1} is due to \citet{stute93TheStrong}, and item \ref{prop:supConverSH2} is the Glivenko-Cantelli theorem.

\begin{proposition}\label{prop:ProbBounds}
Let $\beta \in (0,1)$, then
\begin{enumerate}[i)]
\item $\Prob(\widehat S(t) \leq \beta^{-1}S(t),\quad \forall t \leq \tau_n) \geq 1-\beta,$
\label{prop:ProbBound1}
\item $\Prob\left(Y(t)/n \leq \beta^{-1}(1-H(t-)),\quad\forall t\leq\tau_n\right)\geq1-\beta,$ and \label{prop:ProbBound2}
\item $\Prob\left(Y(t)/n \geq \beta(1-H(t-)),\quad\forall t\leq\tau_n\right)\geq1-e(1/\beta) e^{-1/\beta}.$ \label{prop:ProbBound3} 
\end{enumerate}
i.e., i) $\sup_{t\leq \tau_n} \widehat S(t)/S(t) = O_p(1)$, ii) and  iii) $\sup_{t\leq \tau_n}Y(t)/(n(1-H(t-)) = \Theta_p(1)$.
\end{proposition}
The proofs of items \ref{prop:ProbBound1} and \ref{prop:ProbBound3} are due \citet{gill1983large}, and  item \ref{prop:ProbBound2} follows from \citet[Theorem 3.2.1]{gill1980censoring}.

\begin{remark}\label{remark:limits}
Notice that the results of the previous propositions still hold if we replace $Y(t)$, $H(t)$ and $n$  by $Y_c(t), H_c(t)$ and $n_c$, respectively.
\end{remark}

\subsection{Double Martingales}\label{sec:specialMartingales} 
In our proofs, we will frequently encounter double martingale integral processes of the form:
\begin{align}
\int_{0}^t\int_0^t h_n(x,y)dW_n(x)dW_n(y), \quad t\geq 0,\label{eqn:DMprocss}    
\end{align}
where $h_n:\R^2\to \R$ is a sequence of symmetric positive-definite random functions, and $(W_n(t))_{t\geq0}$ is a sequence of $(\mathcal{F}_t)$-martingales. The aim of this section is to establish conditions under which such processes converge to zero in probability when evaluated at $t=\tau_n$. This result is formally stated in Theorem~\ref{thm:newGeneralTheorem}, and to prove it, we use some results introduced by \citet{FerRiv2020} for double integrals with respect to martingales.

\begin{definition}
\label{def:predicDsigmaAlgebra}
Define the $\sigma$-algebra  $\mathcal{P}$ on $C=\{(x,y): 0< x <y<\infty\}$ as the $\sigma$-algebra generated by sets of the form
\begin{align*}
(a_1,b_1]\times(a_2,b_2]\times X, \text{ where }0\leq  a_1\leq b_1 < a_2 \leq b_2,\text{ and } X\in \mathcal{F}_{a_2},
\end{align*}
and $\{0\}\times\{0\}\times X$, where $X\in\mathcal{F}_0$.
\end{definition}

 A stochastic process $(h(x,y))_{(x,y)\in C}$ is said to be $\mathcal P$-measurable if it is measurable with respect to $\mathcal P$. 
 
 The following proposition is a simple consequence of the definition of $\mathcal P$.
 \begin{proposition}\label{coll:h predictable}
 Let $k:\R^2 \to \R$ be a measurable function, and let $(h_1(t))_{t\geq 0}$ and $(h_2(t))_{t\geq 0}$ be $(\mathcal F_t)$-predictable stochastic processes. Then the process $(h(x,y))_{(x,y)\in C}$ given by $h(x,y) = k(h_1(x), h_2(y)))$ is $\mathcal P$-measurable.
\end{proposition}

\begin{theorem}\label{thm:basicDoubleMartingale}
Let $h$ be a $\mathcal P$-measurable process, and let $W$ be a right-continuous $(\mathcal F_t)$-martingale. Assume that for all $t \geq 0$, it holds that
\begin{align}
\E \left(\int_{(0,t]} \int_{(0,y)} |h(x,y)||dW(x)||dW(y)|\right)< \infty. \label{eqn:condDMart}
\end{align}

Then, $(Z(t))_{t\geq 0}$, defined by $Z(t)=\int_0^t\int_{(0,y)}h(x,y)dW(x)dW(y)$, is an $(\mathcal{F}_t)$-martingale.
\end{theorem}

In our proofs we are particularly interested in predictable positive-definite processes which are defined as following:
\begin{definition}\label{Def:predictablePDprocess}
 We say a process $(h(x,y))_{x,y\geq 0}$ is a \emph{predictable positive-definite} process if it satisfies the following properties:  i) $h(x,y)=h(y,x)$, ii) $h$ is positive definite, i.e., each realisation of the stochastic process is a positive definite function, iii) $(h(x,y))_{(x,y)\in C}$ is $\mathcal P$-measurable, and iv) $(h(x,x))_{x\geq 0}$ is predictable with respect to $(\mathcal{F}_x)_{x\geq0}$.
 \end{definition}

The next theorem, whose proof is deferred to Section~\ref{appendix:DefeProofs}, gives sufficient conditions under which the process of Equation \eqref{eqn:DMprocss} converges to zero in probability.

\begin{theorem}\label{thm:newGeneralTheorem}
Let $(h_n(x,y))_{x,y\geq 0}$ be a sequence of predictable positive-definite processes, let $W_n$ be a sequence of right-continuous $(\mathcal F_t)$-martingales with predictable and quadratic variation processes denoted by $\langle W_n\rangle$ and $[W_n]$, respectively, and suppose that
\begin{align}
     \E\left( \int_0^{t}\int_0^{t} |h_n(x,y)||dW_n(x)||dW_n(y)|\right)<\infty\label{eqn:conditionTheorem17}
\end{align}
holds for all $n$ large enough, and for all $t \in (0,\tau)$.

Then, if 
\begin{align*}
    \int_0^{\tau_n} h_n(x,x)d\langle W_n\rangle (x) = o_p(1),
\end{align*}
we have that
\begin{align*}
    \int_0^{\tau_n} \int_0^{\tau_n}h_n(x,y) dW_n(x)dW_n(y) = o_p(1),
\end{align*}
and the same holds if we replace $o_p(1)$ by $O_p(1)$.
\end{theorem}
Note that Equation~\eqref{eqn:conditionTheorem17} holds trivially due to the simple nature of our martingales, hence, we will not verify this conditions in our applications of the theorem.

\section{Analysis under the Null Hypothesis: Proof of Theorem~\ref{thm:FinalDistr}}\label{sec:proofMain}
The proof of Theorem~\ref{thm:FinalDistr} is split  into three mains steps:
\begin{itemize}
\item[i)] We find a cleaner asymptotic expression for our test statistic under the null hypothesis. In particular, we show that $Z_n$ can be rewritten as
\begin{align*}
Z_n &=\left(\frac{n}{n_0n_1}\right)^2\sum_{i=1}^n\sum_{j=1}^n \int_0^{\tau_n}\int_0^{\tau_n} K(\widehat F(x-),\widehat F(y-))\frac{(-1)^{c_i+c_j}L(x)L(y)}{Y_{c_i}(x)Y_{c_j}(y)}dM^i(x) dM^j(y).
\end{align*}
Then, by using Theorem \ref{thm:newGeneralTheorem}, we prove that $\widehat{F}(x-)$, $n/(n_0n_1)L(x)$, and $Y_{c}(x)/n$ can be replaced by their respective limits, $F_0(x)$, $\psi(x)S_0(x)$, and $(1-H_{c})(x)\eta^{c-1}(1-\eta)^{-c}$, up to a small additive term that decreases to zero in probability, obtaining that
\begin{align}\label{eqn:nZnfinalform1}
Z_n = \frac{1}{n^2}\sum_{i}^n\sum_{j=1}^n \int_0^{\tau_n}\int_0^{\tau_n} \frac{K(F_0(x),F_0(y))\psi(x)\psi(y)}{(1-G_{c_i}(x))(1-G_{c_j}(y))}\frac{1}{\eta^2}\left(\frac{-\eta}{(1-\eta)}\right)^{c_i+c_j}dM^i(x) dM^j(y)+o_p\left(n^{-1} \right).
\end{align}
\item[ii)]  We prove that the deterministic and random covariates models are asymptotically equivalent in the sense that our test statistic $nZ_n$ has the same asymptotic distribution (when it exists) under both models.

\item[iii)] We obtain the limit distribution of $nZ_n$ under the random covariates model. Our results translate to the deterministic covariates model by using the result of the previous item.
\end{itemize}

\subsection{Step i: Finding a simpler asymptotic representation}\label{sec:nicerRep}

For this step, we work under the deterministic covariates model, but notice that our analysis can be extended to the random covariates model by conditioning on the number of random covariates with value 0, say $N_0$, and by noticing that $N_0/n \to \eta$ almost surely.

From Theorem~\ref{thm:ZndoubleIntegral} and Lemma~\ref{lemma:LRembed}, we have that
\begin{align}
Z_n &=\left\|\frac{n}{n_0n_1}\int_0^{\tau_n}K(\widehat F(x-), \cdot)L(x)d\Lambda^*(x)\right\|_{\mathcal{H}}^2\\
&=\left(\frac{n}{n_0n_1}\right)^2 \int_0^{\tau_n}\int_0^{\tau_n} K(\widehat F(x-), \widehat F(y-)) L(x)L(y)d\Lambda^*(x) d\Lambda^*(y),\label{randombjf81}
\end{align}
where $d\Lambda^*(x)=d\widehat\Lambda_0(x)-d\widehat\Lambda_1(x)$, and
\begin{align}
d\Lambda^*(x)=\frac{dM_0(x)}{Y_0(x)}-\frac{dM_1(x)}{Y_1(x)} = \sum_{i=1}^n (-1)^{c_i}\frac{dM^i(x)}{Y_{c_i}(x)}\label{eqn:propertydLstar}
\end{align}
holds under the null hypothesis. In particular, notice that, under the null hypothesis, $(\Lambda^{*}(x))_{x\geq 0}$ is an $(\mathcal{F}_x)$-martingale with predictable and quadratic variation processes given by 
\begin{align}\label{eqn:propertydLstarCompe}
d\langle \Lambda^* \rangle(x)= \frac{d\Lambda_0(x)}{L(x)} \quad \text{ and }\quad d[\Lambda^*](x)= \frac{dN_0(x)}{Y_{0}(x)^2}+\frac{dN_1(x)}{Y_{1}(x)^2},
\end{align}
respectively. Then,  by substituting  Equation~\eqref{eqn:propertydLstar} in Equation~\eqref{randombjf81}, we can rewrite our test statistic as
\begin{align*}
Z_n &=\left(\frac{n}{n_0n_1}\right)^2\sum_{i=1}^n\sum_{j=1}^n \int_0^{\tau_n}\int_0^{\tau_n} K(\widehat F(x-),\widehat F(y-))L(x)L(y)(-1)^{c_i+c_j}\frac{dM^i(x) dM^j(y)}{Y_{c_i}(x)Y_{c_j}(y)}.
\end{align*}

The main result of this section is the following:
\begin{theorem}\label{thm:cleaning}
Assume Conditions~\ref{con:conditions} and \ref{con:ExtraConditions} hold. Then, under the null hypothesis, it holds that
\begin{align*}
nZ_n = \frac{1}{n}\sum_{i=1}^n\sum_{j=1}^n \int_0^{\tau_n}\int_0^{\tau_n} \frac{K(F_0(x),F_0(y))\psi(x)\psi(y)}{(1-G_{c_i}(x))(1-G_{c_j}(y))}\eta^{-2}\left(\frac{-\eta}{(1-\eta)}\right)^{c_i+c_j}dM^i(x) dM^j(y)+o_p\left(1 \right).
\end{align*}
\end{theorem}

We split the proof of Theorem \ref{thm:cleaning} into three parts: 
\begin{enumerate}
\item[1)] We prove that the pooled Kaplan-Meier estimator, $\widehat{F}(x-)$ can be replaced by its limit $F_0(x)$ (recall $F_0$ is continuous, and thus $F_0(x) = F_0(x-))$,
\item[2)] we prove that $nL(x)/(n_0n_1)$ can be replaced by its limit $\psi(x)S_0(x)$, and 
\item[3)] we prove that $Y_0(x)/n$ and $Y_1(x)/n$ can be replaced by their limits, $\eta (1-H_0(x))$  and $(1-\eta) (1-H_1(x))$, respectively.
\end{enumerate}

For the proof of Theorem~\ref{thm:cleaning} we will recurrently use the following fact: let $q:\R\to [0,1)$ and let $\mu$ be a measure on $[0,\infty)$, then
\begin{align*}
    \left\|\int_0^{\infty} K(q(x),\cdot)\mu(dx)\right\|_{\mathcal H}^2 = \int_0^{\infty}\int_0^{\infty}K(q(x),q(y))\mu(dx)\mu(dy),
\end{align*}
which is a straightforward consequence of the linearity of the inner product and integration, and the reproducing property (this argument was used in the proof of Theorem~\ref{thm:ZndoubleIntegral}).

\subsection{Proof of Theorem~\ref{thm:cleaning}}
\textbf{Part 1): Replacement of $\widehat{F}(x-)$ by $F_0(x)$}\\
The following Lemma proves that $K(\widehat{F}(x-),\widehat{F}(y-))$ can be replaced by $K(F_0(x),F_0(y))$, up to a small error.
\begin{lemma}\label{lemma:cleaning1}
Assume Conditions \ref{con:conditions} and  \ref{con:ExtraConditions}. Then, under the null hypothesis, it holds that
\begin{align*}
nZ_n &= n\left(\frac{n}{n_0n_1}\right)^2 \int_0^{\tau_n}\int_0^{\tau_n} K(F_0(x),F_0(y)) L(x)L(y)d\Lambda^*(x) d\Lambda^*(y)+o_p(1).
\end{align*}
\end{lemma}
\begin{proof}

Using norm notation, the desired result is equivalent to show that
\begin{align}
    n\left\|\left(\frac{n}{n_0n_1}\right)\int_0^{\tau_n} K(\widehat F(x-),\cdot)L(x)d\Lambda^*(x)\right\|_{\mathcal H}^2 =    n\left\|\left(\frac{n}{n_0n_1}\right)\int_0^{\tau_n} K( F(x),\cdot)L(x)d\Lambda^*(x)\right\|_{\mathcal H}^2 + o_p(1).
\end{align}
By triangular inequality $\|b\|_{\mathcal H}-\|a-b\|_{\mathcal H}\leq \|a\|_{\mathcal H}\leq \|b\|_{\mathcal H}+\|a-b\|_{\mathcal H}$, then, by taking square, we deduce that we just need to prove that
\begin{align}
n\left\|\frac{n}{n_0n_1}\int_0^{\tau_n}(K(F(x),\cdot)-K(\widehat F(x-),\cdot))L(x)d\Lambda^*(x)\right\|^2_{\mathcal{H}}&=o_p(1)\label{eqn:lemma20step1},
\end{align}
and that 
\begin{align}
n\left\|\frac{n}{n_0n_1}\int_0^{\tau_n}K(F(x),\cdot)L(x)d\Lambda^*(x)\right\|^2_{\mathcal{H}}=O_p(1)\label{eqn:lemma20step2}.
\end{align}

We begin by verifying Equation~\eqref{eqn:lemma20step2}. Expanding the inner product expression, the left-hand side equals
\begin{align*}
    O(1)\frac{1}{n} \int_0^{\tau_n}\int_0^{\tau_n}K(F(x),F(y))L(x)L(y)d\Lambda^*(x)d\Lambda^*(y).
\end{align*}
Also, notice that $(K(F(x),F(y))L(x)L(y))_{x,y\geq 0}$ is a predictable positive-definite process (recall Definition~\ref{Def:predictablePDprocess}), and since $(\Lambda^*(x))_{x\geq 0}$ is an $(\mathcal{F}_x)$-martingale with predictable and quadratic variation processes given in Equation \eqref{eqn:propertydLstarCompe}, a straightforward application of Theorem~\ref{thm:newGeneralTheorem} tell us that we just need to verify that
\begin{align*}
    \frac{1}{n}\int_0^{\tau_n} K(F(x),F(x))L(x)d\Lambda_0(x)= O_p(1).
\end{align*}
The previous equation holds true by Condition~\ref{con:conditions}, and by using that  $L(x)/n = O_p(1)S_0(x)\psi(x)$ uniformly for all $x\leq \tau_n$ due to Propositions~\ref{prop:ProbBounds}.\ref{prop:ProbBound2} and \ref{prop:ProbBounds}.\ref{prop:ProbBound3}, and since $\psi(x)\leq \eta^{-1}(1-G_0(x))$.

We continue verifying Equation~\eqref{eqn:lemma20step1}. Notice that
\begin{align}
n\left\|\frac{n}{n_0n_1}\int_0^{\tau_n}(K(F(x),\cdot)-K(\widehat F(x-),\cdot))L(x)d\Lambda^*(x)\right\|^2_{\mathcal{H}}=    n\left(\frac{n}{n_0n_1}\right)^2 \int_0^{\tau_n}\int_0^{\tau_n} h(x,y)d\Lambda^*(x) d\Lambda^*(y)\label{eqn:Gotozero1},
\end{align}
where 
\begin{align*}
    h(x,y)&=\left(K(F_0(x),F_0(y))-K(F_0(x),\widehat F(y-))-K(\widehat F(x-),F_0(y))+K(\widehat F(x-),\widehat F(y-))\right)L(x)L(y).
\end{align*}
It is easy to verify that $(h(x,y))_{x,y\geq0}$ is a predictable positive-definite process, then, by using Theorem \ref{thm:newGeneralTheorem}, the desired result follows from proving $n\left(\frac{n}{n_0n_1}\right)^2\int_0^{\tau_n} h(x,x)d\langle \Lambda^*\rangle(x) = o_p(1)$. 

Using that  $L(x)/n = O_p(1)S_0(x)\psi(x)$ uniformly for all $x\leq \tau_n$ and $\psi(x)\leq \eta^{-1}(1-G_0(x))$,  we get
\begin{align*}
&n\left(\frac{n}{n_0n_1}\right)^2\int_0^{\tau_n} h(x,x)d\langle \Lambda^*\rangle(x)=O_p(1)\int_0^{\tau} \omega(x)(1-G_0(x))dF_0(x), 
\end{align*}
where  $\omega(x) = K(F_0(x),F_0(x))-2K(F_0(x),\widehat F(x-))+K(\widehat F(x-),\widehat F(x-)))$. Let $\varepsilon>0$, then for any $t\in (0,\tau)$,
\begin{align}
    &\Prob\left(\int_0^{\tau} \omega(x)(1-G_0(x))dF_0(x)\geq \varepsilon \right) \nonumber\\
    &\quad  \leq \Prob\left(\int_0^{t} \omega(x)(1-G_0(x))dF_0(x)\geq \frac{\varepsilon}{2} \right)+ \Prob\left(\int_t^{\tau} \omega(x)(1-G_0(x))dF_0(x)\geq \frac{\varepsilon}{2} \right)\label{eqn:equationDecom2Prob}
\end{align}
We will prove that both terms on the right-hand side of the previous equation tends to 0 as $n$ approaches infinity. For the first term, notice that, since $\sup_{x\leq t}|\widehat F(x-)- F(x)|=0$ a.s. by Proposition \ref{prop:supConverSH}.\ref{prop:supConverSH1}, and since $K$ is continuous in $[0,1)^2$, it exists $N$, large enough such that for every $n\geq N$, $\omega(x)\leq \varepsilon/4$ uniformly on $[0,t]$, thus $\int_0^{t}\omega(x)(1-G_0(x))dF_0(x) \leq\varepsilon/4$. Therefore, for any $t <\tau$,
\begin{align}
   \limsup_{n\to \infty} \Prob\left(\int_0^{t} \omega(x)(1-G_0(x))dF_0(x)\geq \frac{\varepsilon}{2} \right) = 0,\label{eqn:firstpartDecomProb}
\end{align}

For the second term of the right-hand side of Equation~\eqref{eqn:equationDecom2Prob}, note that  $2K(F_0(x),\widehat F(x-))\leq K(F_0(x),F_0(x))+K(\widehat F(x-)),\widehat F(x-))$,  then $\omega(x)\leq 2K(F_0(x),F_0(x))+2K(\widehat F(x-)),\widehat F(x-))$, and thus we deduce
\begin{align}
    &\Prob\left(\int_t^{\tau} \omega(x)(1-G_0(x))dF_0(x)\geq \frac{\varepsilon}{2} \right)\nonumber\\
    &\quad \quad\leq \Prob\left(\int_t^{\tau}K(F_0(x),F_0(x))(1-G_0(x))dF_0(x)\geq \frac{\varepsilon}{8} \right)+\Prob\left(\int_t^{\tau} K(\widehat F(x-),\widehat F(x-))(1-G_0(x))dF_0(x)\geq \frac{\varepsilon}{8} \right)\label{eqn:secondpartDecomProb}
\end{align}
By Condition~\ref{con:conditions} we deduce that it exists $t'<\tau$, $\int_{t'}^{\tau}K(F_0(x),F_0(x)(1-G_0(x))dF_0(x)<\varepsilon/16$. Then, by combining Equations~\eqref{eqn:firstpartDecomProb} and~\eqref{eqn:secondpartDecomProb},  we get that for every $t\geq t'$
\begin{align*}
  &  \limsup_{n\to \infty} \Prob\left(\int_0^{\tau} \omega(x)(1-G_0(x))dF_0(x)\geq \varepsilon \right) \leq  \limsup_{n\to\infty} \Prob\left(\int_t^{\tau} K(\widehat F(x-),\widehat F(x-)(1-G_0(x))dF_0(x)\geq \frac{\varepsilon}{8} \right),
\end{align*}
and by taking $t\to\tau$, we deduce that $\limsup_{n\to \infty} \Prob\left(\int_0^{\tau} \omega(x)(1-G_0(x))dF_0(x)\geq \varepsilon \right) = 0$ by Condition~\eqref{con:ExtraConditions}, completing our proof.
\end{proof}
\noindent
\textbf{Part 2): Replacement of $\frac{n}{n_0n_1}L(x)$ by $\psi(x)S_0(x)$}\\

By the previous part, our test statistic satisfies
\begin{align*}
nZ_n &= n\left(\frac{n}{n_0n_1}\right)^2 \int_0^{\tau_n}\int_0^{\tau_n} K(F_0(x),F_0(y)) L(x)L(y)d\Lambda^*(x) d\Lambda^*(y)+o_p(1).
\end{align*}
The next step is to show that we can replace $n/(n_0n_1)L(x)$ by its limit $\psi(x)S_0(x)$ without altering the value of $nZ_n$ by more than an error of order $o_p(1)$. We formalise this result in the following lemma.

\begin{lemma}\label{lemma:cleaning2}
Under Conditions \ref{con:conditions} and  \ref{con:ExtraConditions}, it holds
\begin{align*}
nZ_n = n \int_0^{\tau_n}\int_0^{\tau_n} K( F(x),  F(y)) \psi(x)\psi(y)S_0(x)S_0(y)d\Lambda^*(x)d\Lambda^*(y)+o_p(1).    
\end{align*}
\end{lemma}

\begin{proof}
Similar to the proof of Lemma~\ref{lemma:cleaning1}, we just need to prove that 
\begin{align}\label{eqn:cleaning2eq1}
    n\left\|\int_0^{\tau_n} K(F(x),\cdot)\left(\frac{n}{n_0n_1}L(x)-\psi(x)S_0(x)\right)d\Lambda^*(x)\right\|_{\mathcal H}^2
    &=o_p(1),
\end{align}
and
\begin{align}\label{eqn:cleaning2eq2}
    n\left\|\int_0^{\tau_n} K(F(x),\cdot)\psi(x)S_0(x)d\Lambda^*(x)\right\|_{\mathcal H}^2
    &=O_p(1).
\end{align}
We begin by proving Equation~\eqref{eqn:cleaning2eq1}. Note that, by expanding the inner product, it is enough to prove that
\begin{align*}
    n\int_0^{\tau_n}\int_0^{\tau_n}h(x,y) d\Lambda^*(x)d\Lambda^*(y)=o_p(1),
\end{align*}
where
\begin{align*}
    h(x,y)&=K(F(x),F(y))\left(\frac{n}{n_0n_1}L(x)-\psi(x)S_0(x)\right)\left(\frac{n}{n_0n_1}L(y)-\psi(y)S_0(y)\right).
\end{align*}
Also, observe that $(h(x,y))_{x,y\geq 0}$ is a predictable positive-definite process. Then, by Theorem~\ref{thm:newGeneralTheorem}, we just need to verify that $n\int_0^{\tau_n}h(x,y)d\langle \Lambda^*\rangle(x)=o_p(1),$ which is deduced from the following equalities:
\begin{align*}
n\int_0^{\tau_n}h(x,y)d\langle \Lambda^*\rangle(x)&=n\int_0^{\tau_n}K(F(x),F(x))\left(1-\frac{\psi(x)S_0(x)}{n/(n_0n_1)L(x)}\right)^2 \left(\frac{n}{n_0n_1}\right)^2L(x)d\Lambda_0(x)\\
&=O_p(1)\int_0^{\tau}K(F(x),F(x))\left(1-\frac{\psi(x)S_0(x)}{n/(n_0n_1)L(x)}\right)^2\psi(x)dF_0(x)\\
&=o_p(1)
\end{align*}
where the second equality is due to Propositions~\ref{prop:ProbBounds}.\ref{prop:ProbBound2} and \ref{prop:ProbBounds}.\ref{prop:ProbBound3}, from which we deduce that  $L(x)/n= \Theta_p(1)\psi(x)S_0(x)$ uniformly for all $x\leq \tau_n$, and the third equality is due to an application of the dominated convergence theorem in sets of arbitrarily large probability: observe that, by Proposition~\ref{prop:supConverSH}, we have that $\left(1-\frac{\psi(x)S_0(x)}{n/(n_0n_1)L(x)}\right)\to 0$ for all $x<\tau$. Additionally, 
\begin{align*}
K(F(x),F(x))\left(1-\frac{\psi(x)S_0(x)}{n/(n_0n_1)L(x)}\right)^2\psi(x)&=O_p(1)K(F(x),F(x))\psi(x)dF_0(x)\\
&=O_p(1)K(F(x),F(x))(1-G_0(x)),
\end{align*}
since $\psi(x)\leq (1-G_0(x))$ for all $x<\tau$, and note that $K(F(x),F(x))(1-G_0(x))$ is integrable by Condition~\ref{con:conditions}. With these ingredients we can apply the dominated convergence theorem in set of arbitrarity large probability. 

To check Equation~\eqref{eqn:cleaning2eq2}, we follow the same steps, using that $n\int_0^{\tau_n}K(F(x),F(x))\psi(x)^2S_0(x)^2 /L_0(x)d\Lambda_0(x) = O_p(1)$, since $L_0(x)/n = \Theta_p(1)\psi(x)S_0(x)$ uniformly for all $x\leq \tau_n$.
\end{proof}

\noindent
\textbf{Part 3): Replacement of $Y_{0}$ by $\eta n(1-H_{0})$ and $Y_1$ by $(1-\eta) n(1-H_{1})$ }
From the previous step, it holds that
\begin{align*}
nZ_n
&= n\left\|\int_0^{\tau_n} K(F_0(x),\cdot)\psi(x)S_0(x) d\Lambda^* (x)\right\|_{\mathcal H}^2+o_p(1)\nonumber\\
&= \frac{1}{n}\left\|\int_0^{\tau_n} K(F_0(x),\cdot)\psi(x)S_0(x)\left( \frac{dM_0(x)}{Y_0(x)/n}-\frac{dM_1(x)}{Y_1(x)/n}\right)\right\|_{\mathcal H}^2+o_p(1)\nonumber\\
\end{align*}

Our next step is to replace $Y_0/n$ and $Y_1/n$ by their corresponding limits $\eta S_0(1-G_0)$ and $(1-\eta)S_0(1-G_1)$. The following Lemma implies the desired result.

\begin{lemma}
Under Conditions \ref{con:conditions}, it holds
\begin{align}
    \frac{1}{n}\left\|\int_0^{\tau_n} K(F_0(x),\cdot)\psi(x)S_0(x)\left(\frac{1}{(Y_c(x)/n)}-\frac{\eta^{c-1}(1-\eta)^{-c}}{S_0(x)(1-G_c(x))} \right)dM_c(x)\right\|_{\mathcal H}^2 = o_p(1),\label{eqn:cleaning3eq1}
\end{align}
and 
\begin{align}
    \frac{1}{n}\left\|\int_0^{\tau_n} K(F_0(x),\cdot)\frac{\psi(x)}{(1-G_c(x))} dM_c(x)\right\|_{\mathcal H}^2 = O_p(1),\label{eqn:cleaning3eq2}
\end{align}
for any class label $c\in\{0,1\}$.
\end{lemma}
\begin{proof}
We only prove the result for $c = 0$ as the proof for $c=1$ follows the same steps. Define  $\sigma_n(x) = \left(\frac{1}{Y_0(x)/n}- \frac{\eta^{-1}}{S_0(x)(1-G_0(x))}\right),$ and notice that
\begin{align}
    &\frac{1}{n}\left\|\int_0^{\tau_n} K(F_0(x),\cdot)\psi(x)S_0(x)\left(\frac{1}{(Y_0(x)/n)}-\frac{\eta^{-1}}{S_0(x)(1-G_0(x))} \right)dM_0(x)\right\|_{\mathcal H}^2 \nonumber\\
    \quad&= \frac{1}{n}\int_0^{\tau_n}\int_0^{\tau_n} K(F_0(x),F_0(y))\psi(x)\psi(y)S_0(x)S_0(y)\sigma_n(x)\sigma_n(x)dM_0(x)dM_0(y) \label{eqn:randomeqn93984}
\end{align}

Define $h(x,y)$ as the integrand in the previous double integral, and notice that $(h(x,y))_{x,y\geq0}$ is a predictable positive-definite process. Also, recall that $M_0(x)$ is an $(\mathcal{F}_x)_{x\geq 0}$-martingale with predictable and quadratic variation processes given by $d\langle M_0 \rangle(x) = Y_0(x)d\Lambda_0(x)$ and $d[M_0](x) = dN(x)$.  Then, by Theorem~\ref{thm:newGeneralTheorem}, we just need to verify that
\begin{align*}
\frac{1}{n}\int_0^{\tau_n}h(x,y)d\langle M_0\rangle(x)=o_p(1).
\end{align*}

By Proposition~\ref{prop:ProbBounds}.\ref{prop:ProbBound2}, we have $Y_0(x)/n = O_p(1)S_0(x)(1-G_0(x))$ uniformly for all $x\leq \tau_n$.
   \begin{align*}
\int_0^{\tau_n}h(x,y)d\langle M_0\rangle(x)=O_p(1) \int_0^{\tau}K(F_0(x),F_0(x))\psi(x)^2 \sigma_n(x)^2\ind_{\{x\leq \tau_n\}} S_0(x)^2(1-G_0(x))dF_0(x).
\end{align*}
To prove that the previous expression is $o_p(1)$ we use the dominated convergence theorem. For such, observe that $\sigma_n(x) \to 0$ a.s. for all $x<\tau$, and that $\sigma_n(x)^2 = O_p(1)(S_0(x)(1-G_0(x))^2$ uniformly for all $x\leq \tau_n$ due to proposition~\ref{prop:ProbBounds}.\ref{prop:ProbBound3}, therefore
\begin{align*}
K(F_0(x),F_0(x))\psi(x)^2 \sigma_n(x)^2 S_0(x)^2(1-G_0(x)) = O_p(1)K(F_0(x),F_0(x))\frac{\psi(x)^2}{1-G_0(x)},
\end{align*}
uniformly for all $x\leq \tau$. Finally, note that $K(F_0(x),F_0(x))\frac{\psi(x)^2}{1-G_0(x)}$ is integrable by Condition \ref{con:conditions}, since $\psi(x)\leq \min\{\eta^{-1}(1-G_1(x)), (1-\eta)^{-1}(1-G_0(x))\}$. The previous analysis yields that we can use the dominated convergence theorem in sets of large probability, concluding that Equation~\eqref{eqn:cleaning3eq1} holds true. Equation~\eqref{eqn:cleaning3eq2} follows from the same arguments, by noticing that $\int_0^{\tau_n} K(F_0(x),F_0(x))\frac{\psi(x)}{(1-G_0(x))}Y_0(x)d\Lambda_0(x) = O_p(1)$.
\end{proof}

\subsection{Step II: Asymptotically equivalence of two models}\label{sec:equivalence}
After applying Theorem~\ref{thm:cleaning} we get

\begin{eqnarray}
Z_n = \frac{1}{n^2}\sum_{i}^n\sum_{j=1}^n \int_0^{\tau_n}\int_0^{\tau_n} \frac{K(F_0(x),F_0(y))\psi(x)\psi(y)}{(1-G_{c_i}(x))(1-G_{c_j}(y))}\eta^{-2}\left(\frac{-\eta}{(1-\eta)}\right)^{c_i+c_j}dM^i(x) dM^j(y)+o_p\left(n^{-1} \right),
\end{eqnarray}
Recall that the previous expression is valid when considering either deterministic or independent Bernoulli$(\eta)$ covariates. In this section we show that asymptotic results obtained using either the deterministic or random covariates models are equivalent.

Let $s_i = (X_i,\Delta_i, c_i)$ for all $i\in[n]$ and define $J:(\R_+\times\{0,1\}\times\{0,1\})^2\to\R$ as
\begin{align}
J(s_i, s_j)
&=\eta^{-2}\left(\frac{-\eta}{(1-\eta)}\right)^{c_i+c_j}\int_0^{X_i}\int_0^{X_j} \frac{ K(F_0(x),F_0(y))\psi(x)\psi(y)}{(1-G_{c_i}(x))(1-G_{c_j}(y))}dM^i(x) dM^j(y).\label{eqn:JVSTAT}
\end{align}
Also, define
\begin{eqnarray}\label{eqn:Vstat}
V_n = \frac{1}{n^2}\sum_{i=1}^n\sum_{j=1}^n J(s_i, s_j), 
\end{eqnarray}
and notice that $V_n$ is the non-negligible part of $Z_n$, that is, $Z_n = V_n + o_p(n^{-1})$ for both, the deterministic and the random covariates models. In this section only, we denote by $Z_n'$ and $Z_n$, respectively, the test statistics under the deterministic and random covariates models.  More generally, we use an apostrophe (e.g., $V_n'$) to denote any term related to the deterministic covariates model. 

\begin{lemma}\label{lemma:coupling}
There exists a coupling of $Z_n$ and $Z_n'$ such that $Z_n-Z_n' = o_p(n^{-1})$. 
\end{lemma}

\begin{proof}

We construct a coupling such that $V_n-V_n' = o(n^{-1})$, which implies the desired result. The coupling is constructed as follows:

\begin{enumerate}
\item For $c \in \{0,1\}$, denote by $\mu_c$ the measure induced by the random pair $(X,\Delta)$ where $X=\min(T,C)$ and $\Delta=\ind_{\{T\leq C\}}$, and $T$ and $C$ are independent random variables such that $T\sim F_c$ and $C\sim G_c$. 

\item Let $\mathcal L_0 =(X_i,\Delta_i, 0)_{i\geq 1}$  and $\mathcal L_1 = (\bar  X_i, \bar \Delta_i, 1)_{i\geq 1}$ be such that all the triples in $\mathcal L_0$ and $\mathcal L_1$ are independent of everything, and $(X_i,\Delta_i) \sim \mu_0$ and $(\bar X_i, \bar \Delta_i) \sim \mu_1$ for all $i\geq 1$. Note that, since the times are sampled from  continuous distributions, the data points are unique almost surely.

\item In the deterministic covariates model we generate the data, namely $\mathcal D'_n$, by choosing the first $n_0$ elements of $\mathcal L_0$ and the first $n_1$ elements of $\mathcal L_1$.
\item In the random covariates model, we generate the data, namely $\mathcal D_n$, by sampling $N_0\sim\text{Binomial}(n,\eta)$ and then choosing the first $N_0$ elements of $\mathcal L_0$ and the first $n-N_0$ elements of $\mathcal L_1$.
\item Compute $V_n'$ and $V_n$ by using the datasets $\mathcal D'_n$ and $\mathcal D_n$, respectively.
\end{enumerate}

Note that the datasets $\mathcal D_n$ and $\mathcal D_n'$ differ in exactly $Q = |N_0-n_0|$ points. For ease of notation, denote by $s_i = (X_i, \Delta_i, c_i)$ the observations in $\mathcal{D}_n$ and by $s_i' = (X_i', \Delta_i', c_i')$ the observations in $\mathcal{D}_n'$. We sort the datasets in such a way that the first $Q$ elements of each dataset are not present in the other, and the remaining $n-Q$ elements in such a way that $s_j = s_j'$ for $j \geq Q+1$.  Then, a simple computation shows that
\begin{align}
V_n - V_n' &=D_n+R_n,
\end{align}
where
\begin{align}
D_n = \frac{1}{n^2}\sum_{i=1}^Q\sum_{j=1}^Q J(s_i,s_j)-J(s_i',s_j')\quad\text{and}\quad R_n = \frac{2}{n^2}\sum_{i=1}^Q \sum_{j=Q+1}^n J(s_i, s_j)-J(s_i', s_j).\label{eqn:randomvei49213}
\end{align}
We continue by proving that $nD_n=o_p(1)$ and $nR_n=o_p(1)$. To this end, we need the following intermediate result, whose proof is deferred to Appendix~\ref{appendix:DefeProofs}. 

\begin{proposition}\label{thm:PropJ}
Let $A:\R_+ \times \R_+ \to \R$ be a symmetric function such that, for every $c,c'\in\{0,1\}$, the following integrals are finite: $\int_{0}^{\tau} |A(x,x)|(1-G_{c}(x))dF_{c}(x)$ and $\int_{0}^{\tau}\int_{0}^{\tau} A(x,y)^2(1-G_{c}(x))(1-G_{c'}(y))dF_{c}(x)dF_{c'}(y)$. Define $J(s_i, s_j)$ by
\begin{align*}
J(s_i, s_j) = \int_{0}^{X_i} \int_{0}^{X_j} A(x,y)dM^i(x) dM^j(y)
\end{align*}
for all $i,j \in [n]$,  and denote by $\tilde \E$ the expectation conditioned on fixed covariate values $c_i$ and $c_j$. Then the following hold:
\begin{enumerate}[i)]
\item $\tilde \E (J(s_i, s_i)) = \int_{0}^{\tau} A(x,x)(1-G_{c_i}(x))dF_{c_i}(x)$,\label{propJ.1}
\item $\tilde \E\left( |J(s_i,s_i)|\right) <\infty$,\label{propJ.2} 
\item for all $i \neq j$, $\tilde \E(J(s_i, s_j)|s_j) = 0$, and\label{propJ.3}
\item for all $i \neq j$, $\tilde \E\left(J(s_i,s_j)^2\right) = \int_{0}^{\tau} \int_{0}^{\tau} A(x,y)^2(1-G_{c_i}(x))(1-G_{c_j}(y))dF_{c_i}(x)dF_{c_j}(y)$.  \label{propJ.4}
\end{enumerate}
\end{proposition}

In our applications of Proposition~\ref{thm:PropJ}, we will choose 
\begin{align*}
A(x,y) = \frac{1}{\eta^2}\left(\frac{-\eta}{1-\eta}\right)^{c_k+c_{\ell}}\frac{ K(F_0(x),F_0(y))\psi(x)\psi(y)}{(1-G_{k}(x))(1-G_{\ell}(y))}.
\end{align*}

We continue by proving $n|D_n|=o_p(1)$.  Observe that, by symmetry,
\begin{align}
n|D_n|&=\frac{1}{n}\left|\sum_{i=1}^Q\sum_{j=1}^QJ(s_i,s_j)-J(s_i',s_j')\right|\nonumber\\
&\leq\frac{1}{n}\left|\sum_{i=1}^QJ(s_i,s_i)-J(s_i',s_i')\right|+\frac{2}{n}\left|\sum_{i=1}^Q\sum_{j=i+1}^QJ(s_i,s_j)-J(s_i',s_j')\right|\label{eq:D_n}.
\end{align}
For the first term in Equation~\eqref{eq:D_n}, notice that by conditioning on $Q$, we obtain 
\begin{align*}
\E\left(\frac{1}{n}\left.\left|\sum_{i=1}^QJ(s_i,s_i)-J(s_i',s_i')\right|\right|Q\right)
&\leq \E \left(\frac{1}{n} \sum_{i=1}^Q |J(s_i,s_i)| \given  Q\right)+\E \left(\frac{1}{n} \sum_{i=1}^Q |J(s_i',s_i')| \given  Q\right)\\
&= \frac{Q}{n}\E(|J(s_1,s_1)||c_1=0)+\frac{Q}{n}\E(|J(s_1',s_1')||c_1'=1),
\end{align*}
where the last equality holds since, without loss of generality, the first $Q$ elements of $\mathcal{D}_n$ are chosen from $\mathcal{L}_0$ and the $Q$ elements of $\mathcal{D}_n'$ are chosen from $\mathcal{L}_1$. Additionally, notice that by Proposition~\ref{thm:PropJ}.\ref{propJ.2},  $\E(|J(s_1,s_1)||c_1=0)<\infty$ and $\E(|J(s_1',s_1')||c_1'=1)<\infty$. 

Also, note that
\begin{align}
\E\left(\frac{Q}{n}\right)^2\leq \E\left( \frac{Q^2}{n^2}\right) \leq 2\E\left(\left|\frac{N_0}{n}-\eta\right|^2 \right)+2\left|\eta -\frac{n_0}{n}\right|^2 = 2\var\left(\frac{N_0}{n^2}\right)+ 2\left|\eta-\frac{n_0}{n}\right|^2 \to 0  \label{eqn:Qto0}
\end{align}
where the last limit holds because $ \Var(N_0) = \eta (1-\eta) n$, and $n_0/n \to \eta$. Then, by using the previous result, we deduce
\begin{align*}
\E\left(\frac{1}{n}\left|\sum_{i=1}^QJ(s_i,s_i)-J(s_i',s_i')\right|\right)&=\E\left(\E\left(\frac{1}{n}\left.\left|\sum_{i=1}^QJ(s_i,s_i)-J(s_i',s_i')\right|\right|Q\right)\right)=O(1)\E\left(\frac{Q}{n}\right)\to0.
\end{align*}

We continue by proving that the second term in Equation~\eqref{eq:D_n} is $o_p(1)$. Let $\varepsilon>0$, then
\begin{align*}
\Prob\left(\frac{2}{n}\left|\sum_{i=1}^{Q-1}\sum_{j=i+1}^QJ(s_i,s_j)-J(s_i',s_j')\right|\geq \varepsilon\right)&=\E\left(\Prob\left(\frac{2}{n}\left|\left.\sum_{i=1}^{Q-1}\sum_{j=i+1}^QJ(s_i,s_j)-J(s_i',s_j')\right|\geq \varepsilon\right|Q\right)\right)\\
&\leq \frac{1}{\varepsilon^2}\E\left(\Var\left(\left.\frac{2}{n}\sum_{i=1}^{Q-1}\sum_{j=i+1}^QJ(s_i,s_j)-J(s_i',s_j')\right|Q\right)\right),
\end{align*}
and 
\begin{align*}
&\Var\left(\left.\frac{2}{n}\sum_{i=1}^{Q-1}\sum_{j=i+1}^QJ(s_i,s_j)-J(s_i',s_j')\right|Q\right)\\
&=\Var\left(\frac{2}{n}\sum_{i=1}^{Q-1}\sum_{j=i+1}^QJ(s_i,s_j) \given Q \right)
+\Var\left(\frac{2}{n}\sum_{i=1}^{Q-1}\sum_{j=i+1}^QJ(s_i,s_j) \given Q \right)\\
&=\frac{4}{n^2}\sum_{i=1}^{Q-1}\sum_{j=i+1}^Q\Var(J(s_i,s_j))+\Var(J(s_i',s_j'))\\
&=\frac{4}{n^2}\sum_{i=1}^{Q-1}\sum_{j=i+1}^{Q}\E(J(s_i,s_j)^2)+\E(J(s_i',s_j')^2),
\end{align*}
the first inequality holds because the variables $(s_i)_{i=1}^Q$ are independent of $(s_i')_{i=1}^{Q}$, given $Q$. The second equality follows from noticing that $\mathbb{C}ov(J(s_i,s_j),J(s_l,s_k))=0$ for all $(i,j) \neq (l,k)$, which is a simple computation that follows straightforwardly from Proposition~\ref{thm:PropJ}.\ref{propJ.3}.  By Proposition~\ref{thm:PropJ}.\ref{propJ.4}, there exists $C'>0$ such that for all $i\neq j$,  $\E(J(s_i,s_j)^2)\leq C'$ and $\E(J(s_i',s_j')^2)\leq C'$, then
\begin{align*}
\Var\left(\left.\frac{2}{n}\sum_{i=1}^{Q-1}\sum_{j=i+1}^QJ(s_i,s_j)-J(s_i',s_j')\right|Q\right)&\leq C\frac{Q^2}{n^2},
\end{align*}
for some $C>0$. We conclude
\begin{align*}
\Prob\left(\frac{2}{n}\left|\sum_{i=1}^{Q-1}\sum_{j=i+1}^QJ(s_i,s_j)-J(s_i',s_j')\right|\geq \varepsilon\right)
&\leq \frac{1}{\varepsilon^2}\E\left(\Var\left(\left.\frac{2}{n}\sum_{i=1}^{Q-1}\sum_{j=i+1}^QJ(s_i,s_j)-J(s_i',s_j')\right|Q\right)\right)\\
&\leq \frac{C}{\epsilon^2}\E(Q^2/n^2)\to0,
\end{align*}
where the last limit follows from Equation~\eqref{eqn:Qto0}.

Finally, we prove $n|R_n|=o_p(1)$. Let $\varepsilon>0$, then
\begin{align*}
\Prob(n|R_n| \geq \varepsilon) =\E\left(\Prob(n|R_n| \geq \varepsilon|Q) \right) \leq \frac{1}{\varepsilon^2}\E\left(\var(nR_n|Q)\right)
\end{align*} 
and
\begin{align*}
\var(nR_n|Q) &=\frac{4}{n^2}\sum_{i=1}^Q\sum_{j=Q+1}^n\var\left(J(s_i,s_j)|Q\right)+ \var\left(J(s_i',s_j)|Q\right)\leq C \frac{Q(n-Q)}{n^2},
\end{align*}
where the previous equality follows by Proposition~\ref{thm:PropJ}.\ref{propJ.3}, from which we deduce that all non-trivial covariances are 0, and the inequality holds for some constant $C>0$ due to Proposition~\ref{thm:PropJ}.\ref{propJ.4}, as indeed,  $\var\left(J(s_i,s_j)|Q\right)$ and $\var\left(J(s_i',s_j|Q)\right)$ are bounded for any $i,j\in[n]$, independently of the value of the covariates.  Since $\E(Q)/n \to 0$ we conclude that $nR_n = o_p(1)$ which, combined  with the fact that $nD_n = o_p(1)$, yields $n(V_n-V_n') = o_p(1)$.

\end{proof}

\subsection{Step III: Limit distribution under the null hypothesis}\label{sec:limitDistribProof}
In this section we prove Theorem~\ref{thm:FinalDistr} for the random covariates model. Recall that by  Lemma~\ref{lemma:coupling}, the result also holds for the deterministic covariates model. 

Recall from Equation~\eqref{eqn:Vstat} that
\begin{align}
\frac{n_0n_1}{n}Z_n&=\frac{n_0n_1}{n}V_n+o_p(1)=\frac{n_0n_1}{n}\frac{1}{n^2}\sum_{i=1}^n\sum_{j=1}^nJ(s_i,s_j)+o_p(1),\label{eqn:relatinZV}
\end{align}
where $J$ is defined in Equation~\eqref{eqn:JVSTAT} and $s_i=(X_i,\Delta_i,c_i)$ are independent and identically distributed random observations as the covariates are random, which is the main advantage of working in the random covariates model. Then, by using the previous expression, the asymptotic distribution of $(n_0n_1)/nZ_n$ can be obtained from a standard application of the theory of V-statistics. In particular, by Proposition~\ref{thm:PropJ}.\ref{propJ.3}, $\E(J(s_i,s_j)|s_j)=0$ for any $i\neq j$, hence $J$ is a so-called \emph{degenerate} V-statistic kernel, therefore Theorem~4.3.2 of \cite{Koroljuk94} yields
\begin{align}
\frac{1}{n}\sum_{i=1}^n\sum_{j=1}^nJ(s_i,s_j)\overset{\mathcal{D}}{\to}\E(J(s_1, s_1)) + \mathcal Y,\label{eqn:random393nfixvn}
\end{align}
where $\mathcal Y=\sum_{i=1}^{\infty} \lambda_i (\xi_i^2-1)$, $(\xi_i)_{i\geq 1}$ are independent and identically distributed standard normal random variables, and $(\lambda_i)_{i\geq 1}$ are the eigenvalues associated to the integral operator $T_J:L^2(X_1,\Delta_1,c_1) \to L^2(X_1,\Delta_1,c_1)$ defined by $(T_J\phi)(s) = \E_{s_1}(J(s,s_1)\phi(s_1))$, where $\E_{s_1}$ means expectation with respect $s_1=(X_1,\Delta_1,c_1)$. 
From Equation~\eqref{eqn:random393nfixvn} we deduce that
\begin{align*}
\frac{n_0n_1}{n}Z_n\overset{\mathcal{D}}{\to}\eta(1-\eta)\E(J(s_1, s_1)) + \eta(1-\eta)\mathcal Y.
\end{align*}

We finish by computing the asymptotic mean and variance of our test statistic.  By using Proposition~\ref{thm:PropJ}.\ref{propJ.1}, a simple computation shows that
\begin{align*}
\E(J(s_1,s_1))&= \frac{1}{\eta} \int_0^\tau \frac{K(F_0(x),F_0(x))\psi(x)^2}{1-G_{0}(x)}dF_0(x)+ \frac{1}{1-\eta} \int_0^\tau \frac{K(F_0(x),F_0(x))\psi(x)^2}{1-G_{1}(x)}dF_0(x),
\end{align*} 
and since $\mathcal Y$  has mean 0, the asymptotic mean of $(n_0n_1)/n Z_n$ is given by
\begin{align*}
\eta(1-\eta) \E(J(s_1,s_1)) = \int_0^\tau K(F_0(x),F_0(x))\psi(x)dF_0(x).
\end{align*}
Finally, from Proposition~\ref{thm:PropJ}.\ref{propJ.4}, we have that 
\begin{align*}
\var(\mathcal Y) = 2\sum_{i=1}^{} \lambda_i^2 = 2\E(J(s_1,s_2)^2) =\frac{2}{(\eta(1-\eta))^2}\int_0^\tau \int_0^\tau K(F_0(x),F_0(y))^2\psi(x)\psi(y)dF_0(x)dF_0(y),
\end{align*}
from which we deduce that the expression for the asymptotic variance is given by
\begin{align*}
\Var(\eta(1-\eta)\mathcal Y)&=2\int_0^\tau \int_0^\tau K(F_0(x),F_0(y))^2\psi(x)\psi(y)dF_0(x)dF_0(y).
\end{align*}

\section{Wild Bootstrap}
By following exactly the same steps used in the proof of Theorem~\ref{thm:FinalDistr}, the Wild Bootstrap test statistic $Z_n^{\mathcal W}$, given in  Equation~\eqref{eqn:wildBoostrapDefi}, can be rewritten as $Z_n^{\mathcal W} = V_n^{\mathcal W} + o_p(n^{-1})$, where $V_n^{\mathcal{W}} = \frac{1}{n^2}\sum_{i=1}^n \sum_{j=1}^n W_iW_j J(s_i, s_j)$, and $J$ is the kernel defined in Equation \eqref{eqn:JVSTAT}.  Under the random covariates model $V_n$ is a degenerate $V$-statistic, and thus, Theorem~\ref{thm:wildbootstrap} is a direct application of Theorem 3.1 of \citet{dehling94random}. This result can be extended to the deterministic  covariates model   by using the same coupling of Lemma~\ref{lemma:coupling} (however, randomness is taken over $\mathcal W$, as we are conditioning on the sequence of data points).

\section{Proof of Theorem~\ref{thm:convergenceAlternative} and Corollary~\ref{thm:obnibus}}\label{sec:appendixAlter}

In this section we mainly work under the alternative hypothesis, therefore we need to consider this fact when computing some limit results. For instance, in this more general setting the limit of $n/(n_0n_1)L(x)$ is slightly different to its limit under the null hypothesis. Indeed, by Proposition \ref{prop:supConverSH}, for every fixed $x$, it holds
\begin{align}\label{eqn:psistarprop1}
\frac{n}{n_0n_1}L(x) \to\psi^*(x)S_1(x)S_0(x)\quad a.s.,
\end{align}
as the number of observations tends to infinity, and by Proposition \ref{prop:ProbBounds}
\begin{align}\label{eqn:psistarprop2}
\frac{n}{n_0n_1}L(x) =O_p(1) \psi^*(x)S_1(x)S_0(x),
\end{align}
uniformly for all $x \leq \tau_n$. In the previous results, $\psi^*(x)$ is the function defined in Equation~\eqref{eqn:psistar} which satisfies the following bound
\begin{align}
\psi^*(x) \leq \min\left \{\frac{1-G_0(x)}{\eta S_1(x)}, \frac{1-G_1(x)}{(1-\eta)S_0(x)}\right\}.\label{eqn:random1239fj448dax}
\end{align}
Other objects that may vary their limits in  this more general setting are the pooled Nelson-Aalen estimator $\widehat \Lambda$ and Kaplan-Meier statistic $\widehat F$, and in general, any statistic constructed using  the survival times $(T_i)_{i\in[n]}$  (as now, $F_0$ and $F_1$ may or may not be equal). Having this in mind, we proceed to prove  Theorem~\ref{thm:convergenceAlternative} and Corollary~\ref{thm:obnibus}.

To proof Theorem~\ref{thm:convergenceAlternative}, we require the following intermediate result:
\begin{lemma}\label{lemma:convergenceAlteLem}
Assume that Conditions~\ref{con:conditions} and \ref{con:ExtraConditions} hold. Then, 
$\|\phi_c^n-\phi_c\|_{\mathcal H} \overset{\Prob}{\to} 0$ for $c \in \{0,1\}$.
\end{lemma}
\begin{proof}

We prove the result for $c=0$, as the proof for $c=1$ is exactly the same. We consider the deterministic covariates model since for the random covariates model we can condition on the number of data points with covariate equal to 0.

Define the intermediate embeddings, $\widehat \phi_0\in \mathcal H$ and $ \phi'_0 \in \mathcal H$, by
\begin{align*}
\widehat \phi_0(\cdot)=\frac{n}{n_0n_1}  \int_0^{\tau_n} K(\widehat F(y-), \cdot)L(y)d\Lambda_0(y),\text{ and}\quad \phi_0'(\cdot)&=\frac{n}{n_0n_1}  \int_0^{\tau_n} K( F(y), \cdot)L(y)d\Lambda_0(y),
\end{align*}
and recall that $\phi_0^n$ and $\phi_0$ (introduced in~Equations \eqref{eqn:Embeddingsn} and \eqref{eqn:fancyEmbedding1}, respectively) are given by
\begin{align*}
\phi_0^n(\cdot) =  \frac{n}{n_0n_1}  \int_0^{\tau_n} K(\widehat F(y-), \cdot)L(y)d\widehat \Lambda_0(y)\quad\text{and}\quad \phi_{0}(\cdot) = \int_{0}^\tau K(F(y),\cdot)\psi^*(y)S_1(y)dF_0(y).
\end{align*}
Notice that in this setting the pooled distribution $F$ is not necessarily equal to $F_0$. 

Clearly 
\begin{align*}
\|\phi_0^n-\phi_0\|_{\mathcal H} \leq \|\phi_0^n - \widehat \phi_0\|_{\mathcal H}+ \|\widehat  \phi_0 - \phi'_0\|_{\mathcal H}+\|\phi_0' - \phi_0\|_{\mathcal H}.
\end{align*}
We will prove that the three terms in the right-hand-side of the previous equation are $o_p(1)$. First we prove $\|\phi_0^n - \widehat \phi_0\|_{\mathcal H}=o_p(1)$. Recall that the Nelson-Aalen estimator satisfies $(d\widehat \Lambda_0 - d\Lambda_0)(x) = dM_0(x)/Y_0(x)$, then
\begin{align}
\|\phi_0^n - \widehat \phi_0\|_{\mathcal H}^2 &= \left(\frac{n}{n_0n_1} \right)^2 \left\| \int_0^{\tau_n}K(\widehat F(x-), \cdot)L(x)\frac{dM_0(x)}{Y_0(x)}\right\|^2_{\mathcal{H}}= \frac{1}{n_0^2}\int_0^{\tau_n}\int_0^{\tau_n} h(x,y)dM_0(x)dM_0(y),\label{eqn:random123fwef4}
\end{align}
where 
\begin{align*}
h(x,y)&=K(\widehat F(x-), \widehat F(y-))\frac{(Y_1(x)/n_1)(Y_1(y)/n_1)}{(Y(x)/n)(Y(y)/n)}.
\end{align*}
We prove that Equation \eqref{eqn:random123fwef4} is $o_p(1)$ by using Theorem \ref{thm:newGeneralTheorem}. Observe that  $h(x,y)$ is symmetric and positive-definite, $(h(x,y))_{(x,y)\in C}$ is $\mathcal{P}$ measurable by Proposition \ref{coll:h predictable}, and $(h(x,x))_{x\geq 0}$ is predictable with respect to $(\mathcal{F}_x)_{x\geq 0}$. Then, by Theorem \ref{thm:newGeneralTheorem}, the result is deduced from checking that $\frac{1}{n_0^2}\int_0^{\tau_n}h(x,x)d\langle M_0\rangle(x) = o_p(1).$
 By using that $Y_1(x)/Y(x)\leq 1$, $n/n_1 = O(1)$, and that $Y_0(x)/n_0 = O_p(1)S(x)(1-G_0(x))$ uniformly for all $x \leq \tau_n$, we have
 \begin{align*}
    \frac{1}{n_0^2}\int_0^{\tau_n}h(x,x)d\langle M_0\rangle(x)&= \frac{1}{n_0^2}\int_0^{\tau_n}K(\widehat F(x-), \widehat F(x-))\frac{(Y_1(x)/n_1)^2}{(Y(x)/n)^2}Y_0(x)d\Lambda_0(x)\nonumber\\
    &= O(1) \frac{1}{n_0}\int_0^{\tau_n}K(\widehat F(x-), \widehat F(x-))(1-G_0(x))dF_0(x).
 \end{align*}

The previous term is tends to 0 since $\int_0^{\tau_n}K(\widehat F(x-), \widehat F(x-))(1-G_0(x))dF_0(x)<\infty$  by Condition~\ref{con:ExtraConditions}.

We continue by proving $\|\widehat \phi_0 - \phi'_0\|_{\mathcal H}=o_p(1)$.  Let $t<\tau$ and $\epsilon>0$, then
\begin{align*}
&\Prob\left(\|\widehat \phi_0 - \phi'_0\|_{\mathcal H}^2>\epsilon\right)\\
&\leq \Prob\left(\left\|\frac{n}{n_0n_1}\int_0^th(x,\cdot)L(x)d\Lambda_0(x)\right\|_{\mathcal H}^2>\epsilon/2\right)+\Prob\left(\left\|\frac{n}{n_0n_1}\int_t^{\tau_n}h(x,\cdot)L(x)d\Lambda_0(x)\right\|_{\mathcal H}^2>\epsilon/2\right)
\end{align*}
where $h(x,\cdot)=(K(\widehat F(x-),\cdot)-K(F(x),\cdot))$. By taking limsup when $n$ grows to infinity, the first probability is equal to zero since $h(x,y) = \langle h(x,\cdot),h(y,\cdot)\rangle_{\mathcal{H}}\to 0 $ uniformly for all $x,y\leq t$. Then,
\begin{align*}
\limsup_{n\to\infty}\Prob\left(\|\widehat \phi_0 - \phi'_0\|_{\mathcal H}^2>\epsilon\right)
&\leq \limsup_{n\to\infty}\Prob\left(\left\|\frac{n}{n_0n_1}\int_t^{\tau_n}h(x,\cdot)L(x)d\Lambda_0(x)\right\|^2_{\mathcal H}>\epsilon/2\right).
\end{align*}
Finally, by taking limit when $t\to\tau$, and by using the definition of $h$, we obtain
\begin{align*}
\limsup_{n\to\infty}\Prob\left(\|\widehat \phi_0 - \phi'_0\|_{\mathcal H}^2>\epsilon\right)
&\leq \lim_{t\to\tau}\limsup_{n\to\infty}\Prob\left(\left\|\frac{n}{n_0n_1}\int_t^{\tau_n}K(\widehat F(x-),\cdot)L(x)d\Lambda_0(x)\right\|_{\mathcal H}>\sqrt{\epsilon/2}/2\right)\\
&+\lim_{t\to\tau}\limsup_{n\to\infty}\Prob\left(\left\|\frac{n}{n_0n_1}\int_t^{\tau_n}K(F(x),\cdot)L(x)d\Lambda_0(x)\right\|_{\mathcal H}>\sqrt{\epsilon/2}/2\right).
\end{align*}
We prove that each of the terms in the right-hand-side of the previous Equation are 0. For the first one, from Propositions~\ref{prop:ProbBounds}.\ref{prop:ProbBound2} and~\ref{prop:ProbBounds}.\ref{prop:ProbBound3}, we have that $(n/(n_0n_1))L(x) =\beta^{-1}\psi^*(x)S_0(x)S_1(x)$ for all $x\leq \tau_n$   with probability at least $1-p(\beta)$ with $p(\beta) \to 0$ as $\beta \to 0$. Then, with probability  at least $1-p(\beta)$
 we have
\begin{align*}
    \left\|\frac{n}{n_0n_1}\int_t^{\tau_n}K(F(x),\cdot)L(x)d\Lambda_0(x)\right\|_{\mathcal H}^2 &\leq \beta^{ -2}\int_t^{\tau_n}\int_t^{\tau_n}K(\widehat F(x-), \widehat F(y-))S_1(x)\psi^*(x)S_1(y)\psi^*(y)dF_0(x)dF_0(y)\nonumber\\
    &\leq 2\beta^{-2}\int_t^{\tau_n}\int_t^{\tau_n}K(\widehat F(x-), \widehat F(x-))S_1(x)\psi^*(x)S_1(y)\psi^*(y)dF_0(x)dF_0(y)\nonumber\\
    &\leq 2\beta^{-2} \int_t^{\tau_n}S_1(y)\psi^*(y)dF_0(y) \int_t^{\tau_n}K(\widehat F(x-), \widehat F(x-))S_1(x)\psi^*(x)dF_0(x)\nonumber\\
    &\leq \frac{2}{\eta^2}\beta^{-2}\int_t^{\tau_n}K(\widehat F(x-), \widehat F(x-))dF_0(x),
\end{align*}
where the second inequality holds by noting that $2K(s,t)\leq K(s,s)+K(t,t)$ for all $s,t$ since $K$ is positive definite, and the fourth inequality holds Equation~\eqref{eqn:random1239fj448dax}. Finally,

\begin{align}
    &\Prob\left(\left\|\frac{n}{n_0n_1}\int_t^{\tau_n}K(\widehat F(x-),\cdot)L(x)d\Lambda_0(x)\right\|_{\mathcal H}>\sqrt{\epsilon/2}/2\right) \nonumber\\
    &\quad \leq p(\beta)+\Prob\left(\int_t^{\tau_n}K(\widehat F(x-), \widehat F(x-))dF_0(x)\geq  \beta^2\eta^2 \sqrt{\epsilon/2}/4\right)
\end{align}
then, by Condition~\ref{con:ExtraConditions}, we conclude that for all $\beta \in (0,1)$
\begin{align}
    \limsup_{t\to\tau}\limsup_{n\to\infty}\Prob\left(\left\|\frac{n}{n_0n_1}\int_t^{\tau_n}K(\widehat F(x-),\cdot)L(x)d\Lambda_0(x)\right\|_{\mathcal H}>\sqrt{\epsilon/2}/2\right)\leq p(\beta),
\end{align}
then taking $\beta \to 0$ yields the desired result. A similar argument shows that  
\begin{align*}
    \lim_{t\to\tau}\limsup_{n\to\infty}\Prob\left(\left\|\frac{n}{n_0n_1}\int_t^{\tau_n}K(F(x),\cdot)L(x)d\Lambda_0(x)\right\|_{\mathcal H}>\sqrt{\epsilon/2}/2\right) = 0,
\end{align*}
concluding that $\|\widehat \phi_0-\phi_0'\|_{\mathcal H}^2 = o_p(1).$

Finally, we prove that $\| \phi_0' - \phi_0\|_{\mathcal H}=o_p(1)$. Observe that by triangular inequality,
\begin{align*}
&\|\phi_0'-\phi_0\|_{\mathcal H}\\
&\leq\left\|\int_0^{\tau_n}K(x,\cdot)\left(\frac{n}{n_0n_1}L(x)-\psi^*(x)S_1(x)S_0(x)\right)d\Lambda_0(x)\right\|_{\mathcal{H}}+\left\|\int_{\tau_n}^\tau K(F(y),\cdot)\psi^*(y)S_1(y)dF_0(y)\right\|_{\mathcal{H}}.
\end{align*}
Observe that
\begin{align*}
\left\|\int_0^{\tau_n}K(x,\cdot)\left(\frac{n}{n_0n_1}L(x)-\psi^*(x)S_1(x)S_0(x)\right)d\Lambda_0(x)\right\|_{\mathcal{H}}^2&=\int_0^{\tau_n} \int_0^{\tau_n} K(x,y)\sigma(x)\sigma(y)\frac{dF_0(x)}{S_0(x)}\frac{dF_0(y)}{S_0(y)}\\
&=o_p(1)
\end{align*}
where $\sigma(x) = \psi^*(x)S_0(x)S_1(x)-n/(n_0n_1)L(x)$.  The second equality follows from the dominated convergence theorem: by Equation~\eqref{eqn:psistarprop1}, we have $\sigma(x) \to 0$ for all $x\leq \tau$ almost surely, also, by Equations \eqref{eqn:psistarprop2} and \eqref{eqn:random1239fj448dax}, it holds
$\sigma(x)/S_0(x) = O_p(1)\psi^*(x)S_1(x) = O_p(1)(1-G_0(x))$ uniformly for all $x \leq \tau_n$. From Condition~\ref{con:conditions} we have that $\int_0^\tau \int_0^\tau |K(F(x),F(y))|(1-G_0(x))(1-G_0(y))dF_0(x)dF_0(y)<\infty$, then we can apply the dominated convergence theorem on sets of probability as high as desired, obtaining that $\| \phi_0' - \phi_0\|_{\mathcal H} = o_p(1)$.

Finally, observe that
\begin{align*}
\left\|\int_{\tau_n}^\tau K(F(y),\cdot)\psi^*(y)S_1(y)dF_0(y)\right\|_{\mathcal{H}}^2&=\int_{\tau_n}^\tau\int_{\tau_n}^\tau K(F(x),F(y))\psi^{*}(x)\psi^{*}(y)S_1(x)S_1(y)dF_0(x)dF_0(y)\nonumber\\
&= O(1)\int_{\tau_n}^\tau\int_{\tau_n}^\tau K(F(x),F(y))(1-G(x))(1-G(y))dF_0(x)dF_0(y)\nonumber\\
&= o_p(1)
\end{align*}
where the second equality is due to Equation~\eqref{eqn:random1239fj448dax}, and the third is due Condition~\ref{con:conditions} together with the fact that $\tau_n \to \tau$.

\end{proof}
\begin{proof}[Proof of Theorem~\ref{thm:convergenceAlternative}]
Notice that $\sqrt{Z_n} = \|\phi_0^n-\phi_1^n\|_{\mathcal H}$. Then, by using the triangle inequality, we deduce
\begin{align*}
\|\phi_0-\phi_1\|-\|\phi_0-\phi_0^n\|_{\mathcal H}-\|\phi_1-\phi_1^n\|_{\mathcal H}  \leq \sqrt{Z_n} \leq \|\phi_0^n-\phi_0\|_{\mathcal H}+\|\phi_0-\phi_1\|_{\mathcal H}+\|\phi_1-\phi_1^n\|_{\mathcal H},
\end{align*} 
hence, a straightforward application of Lemma~\ref{lemma:convergenceAlteLem} deduces $\sqrt{Z_n} \overset{\Prob}{\to} \|\phi_0-\phi_1\|_{\mathcal H},$
and by squaring both sides, we obtain the desired result.
\end{proof}

\begin{proof}[Proof of Corollary~\ref{thm:obnibus}]
Theorem~\ref{thm:convergenceAlternative} yields 
$\sqrt{Z_n} \overset{\Prob}{\to} \|\phi_0-\phi_1\|_{\mathcal H}$, therefore it is enough to prove that $\phi_0 \neq \phi_1$. Recall from Equation \eqref{eqn:fancyEmbedding1} that for $c\in \{0,1\}$, $\phi_c$ is defined by
\begin{align*}
\phi_c(\cdot) = \int_0^{\tau} K(F(y), \cdot)d\nu_c(y)
\end{align*}
where the measures $\nu_c$ is defined in Equation~\eqref{eqn:measuresInf}. Let $F^{-1}(x) = \inf\{t \geq 0: F(t) = x\}$, which is differentiable almost everywhere since $F$ is the cumulative distribution function of a continuous random variable. Then, by performing a change of variables, we obtain
\begin{align*}
\phi_c(\cdot) = \int_0^1 K(x,\cdot)\frac{1}{f(F^{-1}(x))}d\nu_c(F^{-1}(x)),
\end{align*}
where $f$ is the density function associated with $F$. We conclude that $\phi_c$ is the mean kernel embedding (see Equation~\eqref{eqn:defiMKE}) of the measure $\pi_c(\cdot)$ on $[0,1)$ given by $\pi_c(A) = \int_A \frac{1}{f(F^{-1}(x))}d\nu_c(F^{-1}(x))$.

As the kernel $K$ is continuous and c-universal (recall the definition from Section~\ref{sec:RHKSintro}), the mean kernel embedding of finite signed measures is injective, hence, we just need to check that $\pi_0 \neq \pi_1$ and that both measures are finite.

First, we prove that $F_0 \neq F_1$ implies that $\pi_0 \neq \pi_1$. We proceed by contradiction, i.e., we assume that   $\pi_0 = \pi_1$. Since $F$ and $F^{-1}$ are  non-decreasing functions, $\pi_0 = \pi_1$ implies  $\nu_0= \nu_1$. Observe that $\nu_0$ and $\nu_1$ are continuous measures with respective densities $\psi^*(x)S_1(x)f_0(x)$ and $\psi^*(x)S_0(x)f_1(x)$, where $f_0$ and $f_1$ denote the density functions of $F_0$ and $F_1$, respectively. Since $\nu_0 = \nu_1$, their densities are equal. Let $\tau_0 = \sup\{t: S_0(t)>0\}$ and $\tau_1= \sup\{t: S_1(t)>0\}$, and,  without lost of generality, assume that $\tau_0\leq \tau_1$. By hypothesis, for all $x<\tau_0$ we have that $G_0(x)<1$ and $G_1(x)<1$, hence $\psi^*(x) > 0$. Using that $\psi^*(x)S_1(x)f_0(x) = \psi^*(x)S_0(x)f_1(x)$, we deduce that for all $x <\tau_0$, $S_1(x)f_0(x) = S_0(x)f_1(x)$; thus, $\lambda_0(x) = \lambda_1(x)$, and therefore, by the bijection between hazard and density functions, we have $F_0(x) = F_1(x)$. Since $F_0(\tau_0) =1$, we conclude that $F_0(x) = F_1(x)$ for all $x \in \R_+$, which is a contradiction.

Finally, a simple computation verifies that $\nu_c$ is a finite measure. Indeed, for $c=0$, we have $\int_0^{\tau} d\nu_0(x)= \int_0^{\tau} \psi^*(x)S_1(x)f_0(x)dx$, and from Equation~\eqref{eqn:random1239fj448dax}
it holds that $\psi^*(x)S_1(x)f_0(x) \leq C f_0$, where $C>0$, therefore $\nu_0$ is a finite measure. The same holds for $\nu_1$, concluding the result.

\end{proof}

\section{Deferred proofs}\label{appendix:DefeProofs}
\subsection{Proof of Theorem~\ref{thm:newGeneralTheorem}}

To ease notation, we write $h$ and $W$ instead of $h_n$ and $W_n$, respectively. Define the process $R$ by
 \begin{align}
     R(t)&=\int_0^t \int_0^t h(x,y)dW(x)dW(y),\quad t\geq 0.
 \end{align}
 Observe that $R(0) = 0$ and that  $R(t)\geq 0$ for every $t\geq 0$  since $h$ is positive definite. We will prove that exists an increasing predictable process $A$ that majorises $R$ in the sense that $\E(R(T)) \leq  \E(A(T))$ for every finite stopping time, and then, we will apply the Lenglart-Rebolledo inequality to $R$ and $A$, which tells us that for every $\varepsilon, \delta>0$, $\Prob(\sup_{t\leq T} R(t)\geq \varepsilon)\leq \varepsilon/\delta + \Prob(A(T)\geq \delta)$ for any stopping time $T$ (even unbounded).
 
 To find the process $A$, we start by noticing that $R$ can be decomposed as $R(t)=D(t)+2Z(t)$, where the processes  $D$ and $Z$ are given by 
 \begin{align*}
     D(t) &= \int_0^t h(x,x)d[W](x), \text{ and}\quad Z(t)= \int_0^t\int_{(0,y)} h(x,y)dW(x)dW(y),
 \end{align*}
respectively. We make two observations. First, since the process $(h(x,y)_{(x,y)\in C}$ is $\mathcal P$-measurable, then $(Z(t))_{t\geq 0}$ is an $(\mathcal F_t)$-martingale  by Theorem~\ref{thm:basicDoubleMartingale}. Additionally, by the Optional Stopping Theorem, $Z(T) = 0$ for every bounded stopping time $T$. Second, since $h$ is positive definite, $h(t,t)\geq 0$ and $(h(x,x))_{x\geq 0}$ is $(\mathcal F_x)$-predictable, then the process $(D(t))_{t\geq 0}$ is increasing and adapted with $D(0) = 0$ and it is compensated by the process $A$ given by
 \begin{align}
     A(t) = \int_0^t h(Q(x),Q(x))d\langle W \rangle
 \end{align}
 which is increasing and predictable with respect to $(\mathcal F_t)_{t\geq 0}$. Since $D-A$ is a martingale, we have $\E(D(T)) = \E(A(T))$ for any bounded stopping time $T$. We conclude that for every bounded stopping time $T$, it holds that
 $\E(R(T)) = \E(A(T))$, i.e., $A$ majorises $R$. Then, if $A(\tau_n)=o_p(1)$, the Lenglart-Rebolledo inequality concludes that $R(\tau_n) = o_p(1)$, and if $A(\tau_n) = O_p(1)$, then the Lenglart-Rebolledo inequality yields $R(\tau_n) = O_p(1)$.

\subsection{Proof of Proposition~\ref{thm:PropJ}}
$\quad$\\
\textbf{Item \ref{propJ.1}:}
Observe that
\begin{align*}
\bar \E\left( J(s_i, s_i)\right) &= \bar{\E}\left( \int_{0}^{X_i} \int_{0}^{X_i} A(x,y)dM^i(x) dM^i(y)\right)=\bar{\E}(D(X_i))+2\bar{\E}(Z(X_i)),
\end{align*}
 where $D(t)$ and $Z(t)$ are defined by
\begin{align*}
D(t) = \int_0^{t} A(x,x)d[M^i](x)\quad\text{and}\quad Z(t)= \int_0^{t}\int_{(0,y)} A(x,y)dM^i(x)dM^i(y),
\end{align*}
and  notice that $d[M^i](x)=dN^i(x)$ since the survival times are continuous. 

We begin by noticing that
\begin{align*}
\tilde{\E}(D({X_i})) = \tilde \E\left(\int_0^{X_i} A(x,x)dN^{i}(x)\right)=\tilde \E\left(A(X_i,X_i)\Delta_i\right)= \int_{0}^\tau A(x,x)(1-G_{c_i})dF_{c_i}(x).
\end{align*}
Also, by Theorem~\ref{thm:basicDoubleMartingale} and conditioned on the value of the covariate, $Z(t)$ is a zero-mean $(\mathcal{F}_t)$-martingale. Then the optional stopping time theorem yields $\tilde \E(Z(X_i))= \tilde\E (Z(0))=0$, and thus 
\begin{align*}
\bar \E\left( J(s_i, s_i)\right) &=\int_{0}^\tau A(x,x)(1-G_{c_i})dF_{c_i}(x).
\end{align*}

\textbf{Item \ref{propJ.2}:} Observe that 
\begin{align*}
|J(s_i,s_i)|\leq \Delta_i |A(X_i,X_i)|+2\Delta_i\left|\int_0^{X_i}A(X_i,x)d\Lambda_{c_i}(x)\right|+\left|\int_0^{X_i}\int_0^{X_i}A(x,y)d\Lambda_{c_i}(x)d\Lambda_{c_i}(y)\right|.
\end{align*}
We continue by proving that all the terms on the right-hand-side of the previous equation have positive expectation. 

For the first term we have
\begin{align}
\tilde \E\left(\Delta_i |A(X_i,X_i)|\right)= \int_{0}^\tau |A(x,x)|(1-G_{c_i}(x))dF_{c_i}(x)<\infty,\label{eqn:ExpecAbsJ-1}
\end{align}
where the last inequality follows from the assumptions made in the Proposition. The second term satisfies
\begin{align}
\tilde \E\left(\Delta_i\left|\int_0^{X_i} A(X_i,x)d\Lambda_{c_i}(x)\right|\right)\leq  \int_0^{\tau} \int_0^y |A(y,x)|d\Lambda_{c_i}(x)(1-G_{c_i}(y))dF_{c_i}(y)<\infty,
\label{eqn:ExpecAbsJ-2}
\end{align}
where the last inequality follows from Lemma 1 of \citet{Efron1990} (see also Remark 2 of \citet{akritas2000}). Finally, the last term satisfies
\begin{align}
\tilde \E\left(\left|\int_0^{X_i} \int_0^{X_i} A(x,y)d\Lambda_{c_i}(x)d\Lambda_{c_i}(y)\right|\right) &\leq \tilde \E \left(\int_0^\tau \int_0^\tau \ind_{\{X_i \geq \max(x,y)\}} |A(x,y)| d\Lambda_{c_i}(x)d\Lambda_{c_i}(y)\right)\nonumber\\
&=\int_0^\tau \int_0^\tau |A(x,y)|(1-H_{c_i}(x \vee y))d\Lambda_{c_i}(x)d\Lambda_{c_i}(y)\nonumber\\
&=2\int_0^{\tau} \int_0^{y} A(x,y)|(1-G_{c_i}(y))dF_{c_i}(x)d\Lambda_{c_i}(y)\nonumber\\
&<\infty,\nonumber
\end{align}
where the second equality holds by the symmetry of the functions $A(x,y)$ and $1-H_{c_i}(x\vee y)$, and the last inequality holds by Equation \eqref{eqn:ExpecAbsJ-2}.

\textbf{Item \ref{propJ.3}:}  Let $i\neq j$ and define the process
\begin{align}
Q(t) = \int_0^t \int_0^{X_j} A(x,y)dM^j(y)dM^i(x).\label{eqn:Qprocess}
\end{align}
Notice that conditioned on $s_j=(X_j,\Delta_j,c_j)$ and $c_i$, the process $(Q(t))_{t\geq 0}$ is a zero-mean martingale. Therefore, by the optional stopping theorem we have 
\begin{align*}
\tilde \E(J(s_i,s_j)|s_j) = \tilde \E (Q(X_i)|s_j) = \tilde \E(Q(0)|s_j) = 0.
\end{align*}

\textbf{Item \ref{propJ.4}:} Observe that $\tilde{\E}(J(s_i,s_j)^2)=\tilde \E(Q(X_i)^2)$ for $i\neq j$, where $(Q(t))_{t\geq 0}$ is the process defined in Equation~\eqref{eqn:Qprocess}. Observe that, conditioned on $s_j=(X_j,\Delta_j,c_j)$ and $c_i$, $(Q(t)^2)_{t\geq 0}$ is an $(\mathcal{F}_t)$-submartingale with compensator given by
\begin{align*}
\langle Q \rangle(X_i) = \int_0^{X_i} \left(\int_0^{X_j} A(x,y)dM^j(y)\right)^2Y^i(x)d\Lambda_{c_i}(x).
\end{align*}

Then
\begin{align*}
\tilde \E(Q(X_i)^2|s_j)=\tilde \E(\langle Q \rangle(X_i) |s_j)&=\bar{\E}\left(\left.\int_0^{X_i} \left(\int_0^{X_j} A(x,y)dM^j(y)\right)^2Y^i(x)d\Lambda_{c_i}(x)\right|s_j\right)\\
&=\int_0^{\tau} \left(\int_0^{X_j} A(x,y)dM^j(y)\right)^2\E(Y^i(x))d\Lambda_{c_i}(x)\\
&=\int_0^{\tau} \left(\int_0^{X_j} A(x,y)dM^j(y)\right)^2(1-G_{c_i}(x))dF_{c_i}(x)
\end{align*}
and thus
\begin{align*}
\tilde \E(Q(X_i)^2)= \tilde \E(\tilde \E(Q(X_i)^2|s_j)) &= \tilde \E\left(\int_0^{\tau} \left(\int_0^{X_j} A(x,y)dM^j(y)\right)^2(1-G_{c_i}(x))dF_{c_i}(x)\right)\\
&= \int_0^{\tau} \tilde \E\left(\left( \int_0^{X_j} A(x,y)dM^j(y)\right)^2\right)(1-G_{c_i}(x))dF_{c_i}(x)\nonumber\\
&= \int_0^{\tau} \tilde \E\left( \int_0^{\tau} A(x,y)^2Y_j(y)d\Lambda_{c_j}(y)\right) (1-H_{c_i}(x))\Lambda_{c_i}(x) \nonumber\\
&=\int_0^\tau \int_0^\tau A(x,y)^2 (1-G_{c_j}(y))(1-G_{c_i}(x))dF_{c_j}(y)dF_{c_i}(x),
\end{align*}

where the second equality is due to the independence of $s_i$ and $s_j$ for $i\neq j$, and the third equality follows by noticing that, conditioned on $c_j$, $\left(\int_0^t A(x,y)dM^j(y)\right)^2$ is an $(\mathcal{F}_t)$-submartingale for any fixed $x\in\R_+$ and its compensator is given by $\int_0^t A(x,y)^2Y^j(y)d\Lambda_{c_j}(y)$.

\end{document}